\documentclass[11pt]{article}
\usepackage{amsmath, amssymb, amsthm, graphicx}
\usepackage[makeroom]{cancel}
\usepackage{tcolorbox}
\usepackage[colorlinks=true,linkcolor=blue,citecolor=blue]{hyperref}
\usepackage{tikz-cd}
\usepackage{array}
\usepackage{proof}
\usepackage{stmaryrd}
\usepackage{enumerate}
\usepackage{subcaption}

\usetikzlibrary{automata, positioning}

\newcommand{\up}[2][]{\mathcal{W}_{#1}(#2)}

\newcommand{\F}{\mathcal{F}}

\newcommand{\II}{\mathcal{I}}

\renewcommand{\O}{\mathcal{O}}

\newcommand{\U}{\mathsf{U}}

\newcommand{\N}{\mathbb{N}}
\newcommand{\W}{\mathcal{W}}
\newcommand{\Z}{\mathbb{Z}}

\newcommand{\dom}{\mathrm{dom}}

\renewcommand{\int}[1]{\llbracket #1\rrbracket}

\newcommand{\res}[1]{|_{#1}}
\newcommand{\im}[1]{\mathrm{im}(#1)}

\newcommand{\restr}[2]{#1\res{#2}}

\newcommand{\slot}{\raisebox{-.7mm}{\textcolor{gray}{\text{\Large$\blacksquare$}}}}

\newcommand{\card}{\mathsf{Card}}

\newcommand{\Eq}{\mathsf{Eq}}
\newcommand{\Ev}{\mathsf{Ev}}
\newcommand{\Od}{\mathsf{Od}}
\newcommand{\NC}{\mathsf{NC}}
\newcommand{\ND}{\mathsf{ND}}

\newtheorem{theorem}{Theorem}[section]
\newtheorem{lemma}{Lemma}[section]
\newtheorem{corollary}{Corollary}[section]
\newtheorem{proposition}{Proposition}[section]

\theoremstyle{definition}
\newtheorem{definition}{Definition}[section]

\theoremstyle{remark}
\newtheorem{example}{Example}[section]
\newtheorem{remark}{Remark}[section]

\newtheorem{claim}{Claim}

\newcommand{\pic}[2][-.5mm]{\raisebox{#1}{\includegraphics{Images/#2}}}

\title{Local generation of languages}
\author{Mathieu Hoyrup}
\begin{document}
\maketitle
\begin{abstract}
Given a language, which in this article is a set of strings of some fixed length, we study the problem of producing its elements by a procedure in which each position has its own local rule. We introduce a way of measuring how much communication is needed between positions. The communication structure is captured by a simplicial complex whose vertices are the positions and the simplices are the communication channels between positions. The main problem is then to identify the simplicial complexes that can be used to generate a given language. We develop the theory and apply it to a number of languages.
\end{abstract}

\tableofcontents

\section{Introduction}

In this article, a language is a set~$L\subseteq A^I$ where~$A$ and~$I$ are finite sets. Typically, we will take~$A=\{0,1\}$ and~$I_n=[0,n-1]$ for some positive~$n\in \N$.

Our goal is to generate~$L$ by a procedure in which the elements of~$I$ choose values in~$A$ in order to form any element of~$L$, and communicate as little as possible. More precisely, a procedure is a function~$f:B^J\to A^I$ for some finite sets~$B,J$, and it generates~$L$ if its image is~$L$. Each output cell~$i\in I$ only reads the values of the input on a subset~$\W_f(i)\subseteq J$, called its input window. Communication between elements of~$I$ is possible when their input windows intersect, and our goal is to find procedures in which input windows intersect as little as possible.

The communication structure of such a function~$f$ can be captured by a simplicial complex over vertex set~$I$, called the communication complex of~$f$ and written~$K_f$. For instance, there is an edge between two vertices~$i,j\in I$ if~$\W_f(i)\cap\W_f(j)\neq\emptyset$, there is a triangle joining~$i,j,k$ if~$\W_f(i)\cap\W_f(j)\cap\W_f(k)\neq\emptyset$, and so on. One can think of the inputs cells in~$J$ as the simplices of a complex and the output cells as the vertices, reading the values assigned to their incident simplices. 
We say that a simplicial complex~$K$ generates~$L$ if there is a function~$f$ generating~$L$, whose communication complex~$K_f$ is contained in~$K$. Given a language~$L$, we want to identify the minimal complexes that can generate~$L$. 

We introduce the main definitions and study their fundamental properties. We then identify, for several languages, the minimal simplicial complexes generating those languages. In many cases, it turns out they are graphs, but sometimes higher-dimensional complexes are necessary. In several cases, we are able to characterize the minimal complexes generating the languages; in other cases, the problem turns out to be more difficult and we identify some of them but fail to obtain a complete characterization. For instance, the language of binary strings having exactly one occurrence of~$1$ belongs to the second category.

The idea of vertices having local rules determining their values or state depending on the values of their neighborhoods is similar in spirit to what happens in cellular automata \cite{Kari05}, non-uniform cellular automata \cite{DFP12}, or automata networks \cite{GM90,Gad19,GGPT21}. However, in our case the input values are not placed on the vertices, but on the edges or more general simplices. More importantly, those models are usually studied as dynamical systems evolving over time, but we do not allow any form of iteration: the function is to be applied once, so information cannot circulate by successive applications of the rules. In particular, two vertices that are not directly connected by an edge cannot communicate at all, even if they are joined by a path. We also observe that in the theory of automata networks, the input and output cells belong to the same set and the binary relation reflecting the dependency between them is captured by the interaction graph \cite{Gad19}; this graph can be translated in our setting by a simplicial complex whose vertices are the output cells and the simplices are input cells, and a vertex depends on a simplex if it belong to it.

The problem of producing sequences respecting global specifications by local computations is very classical in distributed computing, and has been investigated via many different graph-based models: distributed reactive systems \cite{PnueliR90}, distributed networks \cite{NS95}, distributed graph automata \cite{R15}, distributed environments \cite{FG18}. Simplicial complexes are also used in the study of distributed algorithms via combinatorial topology \cite{HS99,HKR13}. The problem of language generation can be expressed in that framework, but the task that we consider is unusual in distributed computing and it is not clear whether the techniques developed in this area can help solving our problem. We discuss this point in Section \ref{sec_distributed}.

The article \cite{FH24} also explores the idea of generating a language with restrictions on the access that output cells have on the input cells. However, the restrictions studied in \cite{FH24} are not directly about the communication between the output cells, but on the maximal number of input cells that are visible by each one of them. Although these two types of restrictions do have an impact on each other, their precise relationship is not straightforward.

In Section \ref{sec_framework}, we introduce the main concepts and prove their structural properties. In Section \ref{sec_simple_examples} we study several families of languages and obtain in each case a characterization of the minimal complexes generating them. In Section \ref{sec_complicated_examples}, we study two families of languages whose analysis turns out to be more difficult and present partial results. In Section \ref{sec_distributed} we discuss how the language generation problem can be translated in the framework of distributed computing via combinatorial topology. We conclude in Section \ref{sec_conclusion} with possible future directions.


\section{Generation of a language}\label{sec_framework}
In this section we introduce the main concepts of the article.

We start with a few notations. Let~$A,I$ be two sets. If~$x\in A^I$ and~$i\in I$, then we denote~$x(i)$ by~$x_i$. If~$W\subseteq I$ and~$x\in A^I$,~$\restr{x}{W}\in A^W$ is the restriction of~$x$ to~$W$. For~$n\in \N$, let~$I_n=[0,n-1]$. We will usually denote~$A^{I_n}$ by~$A^n$.

\subsection{Visibility}
Let~$A,B$ and~$I,J$ be finite sets and~$f:B^J\to A^I$ be a function. The elements of~$I$ will be called the \textbf{output cells}, and the elements of~$J$ the \textbf{input cells}. The evaluation of~$f$ should be thought as follows: each output cell observes the content of some input cells and applies a local rule to determine its content. We write~$f_i(x)$ for~$f(x)_i$.

The following notion captures the idea that an output cell may only have a partial view on the input, as through a window. It appears in several places with different names, for instance \cite{Gad19,FH24}.
\begin{definition}[Input windows]
Let~$f:B^J\to A^I$. To each~$i\in I$ we associate its \textbf{input window}~$\W_f(i)$, which is the smallest subset~$W$ of~$J$ such that for all~$x,y\in B^J$, if~$\restr{x}{W}=\restr{y}{W}$, then~$f_i(x)=f_i(y)$.
\end{definition}
\begin{proposition}\label{prop_well_defined}
This notion is well-defined, i.e.~there is indeed a smallest such set~$W$.
\end{proposition}
\begin{proof}
There are finitely many sets~$W$ satisfying the condition that~$\restr{x}{W}=\restr{y}{W}$ implies~$f_i(x)=f_i(y)$. Let~$W_0,\ldots,W_k$ be an enumeration of these sets, and let~$\W_f(i)$ be their intersection. Let~$x,y$ coincide on~$\W_f(i)$. We can build a sequence of inputs~$x=x_0,x_1,\ldots,x_{k+1}=y$ such that for each~$j\leq k$,~$x_j$ and~$x_{j+1}$ coincide on~$W_j$, implying that~$f_i(x)=f_i(x_0)=\ldots=f_i(x_{k+1})=f_i(y)$. One simply let~$x_{j+1}$ coincide with~$x$ on~$W_0\cap\ldots \cap W_j$ and with~$y$ elsewhere. The final string~$x_{k+1}$ coincides with~$x$ hence with~$y$ on~$\W_f(i)$, and with~$y$ elsewhere, so~$x_{k+1}=y$.
\end{proof}

It will also be useful to adopt the viewpoint of an input cell and identify the output cells that can see it.
\begin{definition}[Dual windows]
Let~$f:B^J\to A^I$. We define the \textbf{dual window} of~$j\in J$ as
\[
\W^f(j)=\{i\in I:j\in\W_f(i)\}.
\]
\end{definition}

The \textbf{visibility diagram} of~$f$ is the relation
\[
\{(i,j):j\in\W_f(i)\}=\{(i,j):i\in\W^f(j)\}\subseteq I\times J.
\]
It will be be depicted as in Figure \ref{fig_vis_diag}, where a square in column~$i$ and row~$j$ indicates that~$j\in\W_f(i)$, i.e.~that~$i$ has access to the content of~$j$, or equivalently that~$i\in\W^f(j)$, i.e.~that~$j$ is visible by~$i$. In network automata, there is no distinction between input and output cells and the visibility diagram is represented as a graph, called the interaction graph \cite{Gad19}. In our setting, it can be seen as a bipartite graph, but the representation as a boolean matrix is more readable.

\begin{example}
Let~$f:\{0,1\}^{\{a,b,c,d\}}\to \{0,1\}^{\{A,B,C\}}$ be defined by
\[
f(a,b,c,d)=(a\times b,b\times c,c\times d).
\]
Its visibility diagram is shown in Figure \ref{fig_vis_diag} and summarizes the following information:
\[
\begin{array}{lcl}
\W_f(A)=\{a,b\}&&\W^f(a)=\{A\}\\
\W_f(B)=\{b,c\}&&\W^f(b)=\{A,B\}\\
\W_f(C)=\{c,d\}&&\W^f(c)=\{B,C\}\\
&&\W_f(d)=\{C\}
\end{array}
\]

\begin{figure}[ht]
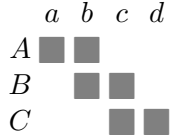

\begingroup
\arraycolsep=1pt
\begin{equation*}
\begin{array}{cccccc}
&a&b&c&d\\
A&\slot&\slot&&\\
B&&\slot&\slot&\\
C&&&\slot&\slot&
\end{array}
\end{equation*}
\endgroup
\caption{The visibility diagram of~$f:\{0,1\}^{\{a,b,c,d\}}\to \{0,1\}^{\{A,B,C\}}$}\label{fig_vis_diag}
\end{figure}
\end{example}

It is convenient to introduce the following notation.
\begin{definition}
For~$S\subseteq I$, let
\begin{align*}
\up[f]{S}&=\bigcap_{i\in S}\W_f(i)\\
&=\{j\in J:S\subseteq \W^f(j)\}.
\end{align*}
\end{definition}
Intuitively,~$\up[f]{S}$ is the set of input cells that are visible by all the output cells in~$S$. The fact that~$\up[f]{S}$ is non-empty can be interpreted as the possibility, for all the output cells in~$S$, to communicate all together in order to choose their contents.

\subsection{Communication complex}
When an input cell is read by several output cells, it allows these ones to communicate in order to choose their content. We introduce an object that captures this communication structure.

We first recall the classical notion of simplicial complex and refer to \cite{Kozlov07} for more information.
\begin{definition}[Simplicial complex]
A \textbf{simplicial complex} (or \textbf{complex}) over a set~$I$ is a collection~$K$ of subsets of~$I$ which is downwards closed: if~$S\subseteq T$ and~$T\in K$, then~$S\in K$. The elements of~$K$ are called \textbf{simplices}.
\end{definition}
A simplicial complex is entirely determined by its maximal simplices. We will often abuse notation and write~$K=\{S_0,\ldots,S_k\}$ where the~$S_i$'s are its maximal simplices: it should be understood that~$K$ is the simplicial complex induced by the~$S_i$'s, i.e.~the collection of subsets of the~$S_i$'s.

One may think of a simplicial complex as a hypergraph, where the elements of~$I$ are the vertices, the pairs are the edges and the larger simplices are hyperedges.

We then arrive at the two most important definitions of the article.
\begin{definition}[Communication complex]
To a function~$f:B^J\to A^I$ we associate its \textbf{communication complex}~$K_f$, which is the simplicial complex over~$I$ defined for~$S\subseteq I$ by:
\begin{equation*}
S\in K_f\iff \up[f]{S}=\bigcap_{i\in S}\W_f(i)\neq\emptyset.
\end{equation*}
\end{definition}
It is indeed a simplicial complex: if~$S\in K_f$ and~$T\subseteq S$, then~$T\in K_f$.

\begin{definition}[Language generation]
Let~$L\subseteq A^I$. We say that a function~$f:B^J\to A^I$ \textbf{generates}~$L$ if~$\im{f}=L$, and that a simplicial complex~$K$ over~$I$ \textbf{generates}~$L$ if there exists a function~$f$ generating~$L$ such that~$K_f\subseteq K$.
\end{definition}
Of course, if~$K\subseteq K'$ and~$K$ generates~$L$, then so does~$K'$. Therefore, our main goal will be to identify the minimal complexes generating a given language~$L$, where minimal is understood in the sense of inclusion between complexes.

\begin{remark}[The trivial generation procedure]
For any~$L\subseteq A^I$, the full complex~$K=\{I\}$ always generates~$L$ because all the cells can collectively agree on the element of~$L$ to be produced. Formally, let~$B=L$,~$J=\{0\}$ and~$f:L^J\to A^I$ be the ``identity'': its input is an element~$x\in L=L^J$, and its output is~$x\in A^I$. There is a single input cell in~$J$, and all the output cells read it. Whatever the communication complex~$K_f$ is, it is contained in the full complex. 
\end{remark}

It will sometimes be convenient to express the communication complex of~$f$ using its dual windows.

\begin{proposition}[Communication complex via dual windows]\label{prop_complex_dual_windows}
Let~$f:B^J\to A^I$. Its communication complex~$K_f$ is the simplicial complex induced by the dual windows of~$f$, i.e.~for every set~$S\subseteq I$,
\[
S\in K_f\iff \exists j\in J,S\subseteq \W^f(j).
\]
\end{proposition}
In particular, every dual window belongs to~$K_f$, and every maximal simplex of~$K_f$ is some dual window.
\begin{proof}
Let~$S\subseteq I$. One has
\begin{align*}
S\in K_f&\iff \bigcap_{i\in S}\W_f(i)\neq\emptyset\\
&\iff \exists j\in J, S\subseteq \W^f(j).\qedhere
\end{align*}
\end{proof}

\begin{remark}[Varying input alphabets]
When designing a generation procedure, it is sometimes convenient to allow different alphabets for the input cells: for every~$j\in J$, we choose a finite set~$B_j$ and consider functions~$f:\prod_j B_j\to A^I$. Such a function can always be converted into a function~$g:B^J\to A^I$ for a sufficiently large alphabet~$B$, generating the same language and whose communication complex is the same as the one that could be defined for~$f$.
\end{remark}

\paragraph{Examples.} We give a few corner case examples.
\begin{example}[Constant function]
Let~$f:B^J\to A^I$ be constant (i.e.~$\im{f}$ is a singleton). One has~$\W_f(i)=\emptyset$ and~$\W^f(j)=\emptyset$ for all~$i\in I$ and~$j\in J$, and~$K_f$ is the empty complex.

Therefore, if~$L\subseteq A^I$ is a singleton, then the only minimal complex generating~$L$ is the empty complex. Intuitively, no communication is needed because each output cell has a predetermined constant value.
\end{example}
\begin{example}[Identity function]
Let~$f:A^I\to A^I$ be the identity, with~$|A|\geq 2$. For each~$i\in I$, one has~$\W_f(i)=\W^f(i)=\{i\}$, and~$K_f$ is the complex made of singletons, i.e.~$K_f$ is a set of vertices with no edge or higher-order simplices between them.

The only minimal complex generating the full language~$L=A^I$ is the complex made of singletons. Intuitively, no communication between output cells is needed, but every cell needs to read some varying input to produce a varying output.
\end{example}

\begin{example}[Constant sequences]\label{ex_constant}
Assume that~$|A|\geq 2$ and let~$L\subseteq A^I$ be the set of constant sequences. We show that the only complex generating~$L$ is the full complex over~$I$. Let~$f$ generate~$L$. As all the functions~$f_i$ are equal,~$W:=\W_f(i)$ does not depend on~$i\in I$. As~$|A|\geq 2$, each~$f_i$ is non-constant so~$W=\W_f(i)$ is non-empty. Therefore,~$\bigcap_{i\in I}\W_f(i)=W\neq\emptyset$, so~$K_f$ is the full complex.
\end{example}

We will see that most languages are generated by intermediate complexes, sometimes by graphs, sometimes by higher-dimensional complexes.

Our general goal is to identify the family of complexes that generate a given language, and more precisely the minimal ones. Certain languages admit a unique minimal generating complex, other languages admit a larger family of generating complexes.

More generally, we want to relate the properties of a language to the properties of the complexes generating it. For instance, the connectivity properties of complexes generating a language gives information on the correlations or dependencies between values that the language puts on positions in~$I$.

We now develop all the fundamental properties of communication complexes and language generation. Most of them are very natural and simple to prove, and constitute a toolbox that will be very useful when investigating particular languages.


\subsection{Composition}
It frequently happens that a language~$L$ can be obtained from another language~$M$ by applying a surjective function~$f:M\to L$. Any generation procedure for~$M$ then induces a generation procedure for~$L$ by composition. We analyze how it transforms the underlying communication complexes.
\begin{definition}[Image of a complex]
Let~$f:B^J\to A^I$ be continuous and~$K$ be a simplicial complex over~$J$. We define the simplicial complex~$f_*(K)$ over~$I$ as the simplicial complex induced by the sets
\[
f_*(S):=\bigcup_{j\in S}\W^f(j)
\]
for all~$S\in K$. In other words, a subset of~$I$ belongs to~$f_*(K)$ if and only if it is contained in one of these sets.
\end{definition}
%

\begin{proposition}[Image of a language]\label{prop_image}
Let~$f:B^J\to A^I$ and~$g:C^H\to B^J$. One has
\[
K_{f\circ g}\subseteq f_*(K_g).
\]

In particular, if a complex~$K$ generates a language~$M\subseteq B^J$, then~$f_*(K)$ generates the language~$f(M)\subseteq A^I$.
\end{proposition}
\begin{proof}
Every simplex of~$K_{f\circ g}$ is contained in~$\W^{f\circ g}(h)$ for some~$h\in H$. One has~$\W^{f\circ g}(h)\subseteq\bigcup_{j\in\W^g(h)}\W^f(j)$ which belongs to~$f_*(K_g)$ as~$\W^g(h)$ belongs to~$K_g$.

Therefore, if~$g$ generates~$M$ and~$K_g\subseteq K$, then~$f\circ g$ generates~$f(M)$ and~$K_{f\circ g}\subseteq f_*(K_g)\subseteq f_*(K)$.
\end{proof}

\paragraph{Projection.}
A particularly simple operation consisting in projecting a language on a subset of coordinates. It easily corresponds, on the side of simplicial complexes, to restricting the complex to a subset of vertices.

For~$J\subseteq I$, we define the \textbf{projection} function~$\pi_J:A^I\to A^J$, sending~$x$ to~$\restr{x}{J}$. If~$K$ is a complex over~$I$, then~$\restr{K}{J}$ is the complex over~$J$ whose simplices are the simplices of~$K$ that are contained in~$J$.
\begin{corollary}[Projection]
Let~$L\subseteq A^I$. If~$K$ generates~$L$, then~$\restr{K}{J}$ generates~$\pi_J(L)$.
\end{corollary}
\begin{proof}
Let~$f=\pi_J$. One has
\[
\W^f(j)=\begin{cases}
\{j\}&\text{if }j\in J,\\
\emptyset&\text{otherwise.}
\end{cases}
\]
so~$f_*(S)=S\cap J$ for~$S\subseteq I$. Therefore,~$f_*(K)=\restr{K}{J}$ generates~$\pi_J(L)$ by Proposition \ref{prop_image}.
\end{proof}

\paragraph{Canonical generation procedure.}
We now show that if a complex~$K$ generates a language~$L$, then there exists a generation procedure working as follows: arbitrary values are assigned to the simplices of~$K$, or just to its maximal simplices, and each vertex~$i\in I$ has a local rule that, upon reading the values of its incident simplices (i.e., the ones that contain~$i$), determines a value for~$i$. Moreover, one can assume that the values that are assigned to the simplices are elements of~$L$, and that if a vertex~$i$ reads the same element~$x$ in all its incident simplices, then~$i$ takes value~$x_i$. In other words, if~$J=\{S_0,\ldots,S_{k-1}\}$ is the set of maximal simplices of~$K$, then this procedure is a function~$f:L^k\to A^I$ such that for all~$j\in I_k$,~$\W^f(j)=S_j$, and for all~$x\in L$,~$f(x,\ldots,x)=x$.

A particularly simple example is when~$K$ is a graph: values are assigned to the edges, and each vertex determines its value depending on the values of its incident edges.

\begin{proposition}[Canonical generation procedure]\label{prop_canonical}
Let~$L\subseteq A^I$ and~$K$ a complex over~$I$ generating~$L$. Let~$B=L$ and~$J$ be the set of maximal simplices of~$K$. There exists a function~$f:B^J\to A^I$ generating~$L$, such that~$K_f\subseteq K$ and~$f(x,\ldots,x)=x$ for all~$x\in L$.
\end{proposition}
\begin{proof}
For more clarity,~$J$ is an initial segment of~$\N$ and~$(S_j)_{j\in J}$ is the family of maximal simplices of~$K$.

We start with some~$g:C^H\to A^I$ generating~$L$ with~$K_g\subseteq K$, for some finite sets~$C$ and~$H$. We define~$\varphi:B^J\to C^H$ such that~$g\circ \varphi$ generates~$L$ and~$K_{g\circ \varphi}\subseteq K$.

For each~$h$,~$\W^g(h)$ is contained in some maximal simplex~$S_{j_h}$ of~$K$ (Proposition \ref{prop_complex_dual_windows}). For each~$x\in L$, there exists~$z(x)\in C^H$ such that~$g(z(x))=x$. We define~$\varphi:B^J\to C^H$ as follows: for~$h\in H$, let
\[
\varphi_h(y)=z_h(y_{j_h}).
\]

First,~$g\circ \varphi$ generates~$L$. One has~$\im{g\circ \varphi}\subseteq\im{g}=L$. For each~$x\in L$, let~$y_j=x$ for all~$j$. One has~$\varphi(y)=z(x)$ so~$g\circ \varphi(y)=x$. Therefore,~$\im{g}=L$.

Note that~$\W_\varphi(h)\subseteq\{j_h\}$, so if~$h\in\W^\varphi(j)$ then~$j_h=j$. Therefore,
\begin{align*}
\W^{g\circ \varphi}(j)&\subseteq \bigcup_{h\in\W^\varphi(j)}\W^g(h)\\
&\subseteq\bigcup_{h\in\W^\varphi(j)}S_{j_h}\\
&\subseteq S_j.
\end{align*}
As a result,~$K_{g\circ\varphi}\subseteq K$.
\end{proof}

In many cases we will find a much more economical procedure, using a small input alphabet such as~$\{0,1\}$.

\subsection{Irreducible components}
The connectivity of a complex reflects a simple property of the language it generates.

\begin{definition}
A language~$L\subseteq A^I$ is \textbf{irreducible} if there is no non-trivial partition~$I=I_0\sqcup I_1$ such that~$L=\pi_{I_0}(L)\times \pi_{I_1}(L)$.
\end{definition}
For instance, the classical operation of concatenating two languages results in a non-irreducible language.

A simplicial complex~$K$ over~$I$ is \textbf{disconnected} if there exists a non-trivial partition~$I=I_0\sqcup I_1$ such that every simplex of~$K$ is contained in~$I_0$ or in~$I_1$. Otherwise, it is \textbf{connected}.

\begin{proposition}[Irreducible vs connected]\label{prop_irreducible}
A language~$L$ is irreducible if and only if every complex generating~$L$ is connected.
\end{proposition}
\begin{proof}
Assume that a disconnected complex~$K$ generates~$L$. Let~$I=I_0\sqcup I_1$ be a non-trivial partition such that every simplex of~$K$ is contained in~$I_0$ or~$I_1$. Let~$f:B^J\to A^I$ be a generation procedure with~$K_f\subseteq K$. Note that if~$i_0\in I_0$ and~$i_1\in I_1$, then~$\W_f(i_0)\cap \W_f(i_1)=\emptyset$. We show that~$L=\pi_{I_0}(L)\times \pi_{I_1}(L)$. Let~$y_0\in\pi_{I_0}(L)$ and~$y_1\in \pi_1(L)$, and let~$x_0,x_1\in B^J$ be such that~$f(x_0)\res{I_0}=y_0$ and~$f(x_1)\res{I_1}=y_1$. Let~$x\in B^j$ coincide with~$x_0$ on~$\bigcup_{i\in I_0}\W_f(i)$ and with~$x_1$ on~$\bigcup_{i\in I_1}\W_f(i)$. It is well-defined as these two sets do not intersect. One has~$f(x)\res{I_0}=y_0$ and~$f(x)\res{I_1}=y_1$, and~$f(x)\in L$. We have shown that~$L=\pi_{I_0}(L)\times \pi_{I_1}(L)$, so~$L$ is not irreducible.

Conversely, assume that~$L$ is not irreducible, i.e.~there exists a non-trivial partition~$I=I_0\sqcup I_1$ such that~$L=\pi_{I_0}(L)\times \pi_{I_1}(L)$. Let~$f_0:B^{J_0}\to A^{I_0}$ and~$f_1:B^{J_1}\to A^{I_1}$ be functions generating~$\pi_{I_0}(L)$ and~$\pi_{I_1}(L)$ respectively (we can indeed assume that they have the same input alphabet~$B$). Let~$f:B^{J_0}\times B^{J_1}\to A^I$ be defined by~$f(x_0,x_1)_i=f_0(x_0)_i$ if~$i\in I_0$, and~$f(x_0,x_1)_i=f_1(x_1)_i$ if~$i\in I_1$. We can see the input space of~$f$ as~$B^{J_0\sqcup J_1}$. The communication complex of~$f$ is the disjoint union of the communication complexes of~$f_0$ and~$f_1$, and is disconnected.
\end{proof}

The proof shows that a partition of~$I$ witnesses the reducibility of~$L$ if and only if it disconnects some complex generating~$L$. 

\begin{corollary}
Let~$I=I_0\sqcup\ldots\sqcup I_k$ and for each~$i\in [0,k]$, let~$L_i\subseteq A^{I_i}$ be irreducible, and~$L=L_0\times\ldots\times L_k$. A simplicial complex~$K$ is minimal generating~$L$ if and only if~$K=K_0\sqcup\ldots\sqcup K_k$ where for each~$j\leq k$,~$K_j$ is a complex over~$I_i$ which is minimal generating~$L_j$.
\end{corollary}
\begin{proof}
In the second part of the proof of Proposition \ref{prop_irreducible}, one could start with any generation procedure~$g$ for~$L$ and use~$f_0=\pi_{I_0}\circ g$ and~$f_1=\pi_{I_1}\circ g$ to generate the two projections of~$L$. The complexes~$K_{f_0}$ and~$K_{f_1}$ are then obtained from~$K_g$ by keeping the simplices of~$K_g$ that are contained in~$I_0$ and~$I_1$ respectively. The final complex~$K_f$, which is the disjoint union of~$K_{f_0}$ and~$K_{f_1}$, is therefore contained in~$K_g$. It shows in particular that if~$K_g$ is minimal generating~$L$, then~$K_f=K_g$ which is therefore already disconnected via the partition~$I=I_0\sqcup I_1$.

Iterating this argument, a minimal complex generating~$L$ must be disconnected along the partition~$I=I_0\sqcup\ldots\sqcup I_k$.

Conversely, if each~$K_j$ is minimal generating~$L_j$, then~$K$ generates~$L$. If~$K'\subseteq K$ generates~$L$, then its restriction to~$I_j$ generates~$L_j$ so it equals~$K_j$. Therefore,~$K'=K$, showing that~$K$ is minimal generating~$L$.
\end{proof}

Every language~$L\subseteq A^I$ can be uniquely factored as a product of irreducible components. In order to understand how to generate~$L$, it is sufficient to understand how to generate its irreducible components. 

%

\subsection{Symmetries}
In this section, we make precise the intuitive idea that some of the symmetries of a language induce symmetries on the family of its generating complexes.

Let~$G$ be a group acting on~$I$. This action induces an action of~$G$ on~$A^I$:
\[
(g\cdot x)_i=x_{g^{-1}\cdot i}\text{ for~$g\in G$ and~$x\in A^I$.}
\]

We say that a language~$L\subseteq A^I$ is \textbf{$G$-invariant} if for all~$w\in L$ and~$g\in G$, one has~$g\cdot w\in L$.

\begin{proposition}[Symmetries]\label{prop_symmetries}
Let~$G$ be a group acting on~$I$ and~$L\subseteq A^I$ be a~$G$-invariant language. If a complex~$K$ generates~$L$, then~$g\cdot K$ generates~$L$, where
\[
g\cdot K=\{g\cdot S:S\in K\}=\{\{g\cdot i:i\in S\}:S\in K\}.
\]
\end{proposition}
\begin{proof}
It is a direct consequence of Proposition \ref{prop_image}. Indeed, for each~$g\in G$, the function~$f:A^I\to A^I$ mapping~$x$ to~$g\cdot x$ has the property that~$\W^f(i)=\{g\cdot i\}$. If~$K$ generates~$L$, then~$f_*(K)=g\cdot K$ generates~$f(L)=L$.
\end{proof}

If~$L$ is~$G$-invariant, then it is not necessarily the case that every complex generating~$L$, or every minimal complex generating~$L$, is~$G$-invariant. For instance, if~$G$ is the whole group of permutations of~$I$, then all the elements of~$I$ play the same role w.r.t.~$L$, but it may be necessary to break this symmetry to design a generation procedure with little communication. We give such an example.

\begin{example}[Symmetry breaking]
Let~$\card_{\leq 1}\subseteq \{0,1\}^I$ be the language of sequences having at most one occurrence of~$1$. It is symmetric in the maximal possible way, because it is invariant under the action of the whole group of permutations of~$I$. However, we will see in Corollary \ref{cor_at_most} that the minimal complexes generating~$\card_{\leq 1}$ are the trees spanning~$I$. As predicted by Proposition \ref{prop_symmetries}, the class of trees is indeed invariant by permutations, but each individual tree is asymmetric because it both has leaves and internal nodes (unless~$|I|\leq 2$).
\end{example}

%
%
%

\subsection{Realizing a complex}
Certain languages admit many different generation procedures relying on different ways of communicating. For other languages, there is only one way, i.e.~one minimal complex generating these languages.


We show that for every complex~$K$, there exists a language~$L$ such that~$K$ is the only minimal complex generating~$L$.
\begin{proposition}
Let~$K$ be a simplicial complex over the finite set~$I$. There exists a finite set~$A$ and a language~$L\subseteq A^I$ such that the complexes generating~$L$ are the complexes containing~$K$.
\end{proposition}
\begin{proof}
For~$i\in I$, let~$K(i)$ be the set of simplices of~$K$ containing~$i$. Let~$A=\{0,1\}^K$ be the set of assignments of boolean values to the simplices of~$K$. If~$h\in A$ and~$i\in I$, then let~$h_i\in A$ coincide with~$h$ on~$K(i)$ and have value~$0$ elsewhere ($h_i$ is essentially the restriction of~$h$ to~$K(i)$), and let~$x_h:I\to A$ map~$i$ to~$h_i$. The sought language is~$L=\{x_h:h\in A\}\subseteq A^I$.

Let~$f:\{0,1\}^K\to A^I$ be the function sending~$h$ to~$x_h$. It is surjective by definition of~$L$ (and continuous because~$K$ is finite). We prove that~$K_f=K$ by showing~$\W_f(i)=K(i)$ for each~$i\in I$. Observe that~$f_i(h)=x_h(i)=h_i$. If~$h$ and~$h'$ coincide on~$K(i)$, then~$f_i(h)=h_i=h'_i=f_i(h')$, so~$\W_f(i)\subseteq K(i)$. For the other inclusion, let~$S\in K(i)$, and let~$h,h'$ coincide everywhere except at~$S$. As~$S\in K(i)$,~$h_i\neq h'_i$ so~$f_i(h)\neq f_i(h')$. Therefore,~$S\in \W_f(i)$.

Therefore, if~$S\subseteq I$, then~$\bigcap_{i\in S}\W_f(i)$ is non-empty if and only if there exists a simplex in~$K$ containing~$S$, which is equivalent to~$S\in K$.

Conversely, let us show that every complex generating~$L$ contains~$K$.

Fix some~$S\in K$ and let~$g:A^I\to \{0,1\}^I$ be defined by~$g_i(x)=x(i)(S)$. As~$g_i(x)$ is defined from~$x(i)$ and can take any binary value, one has~$\W_{g}(i)=\{i\}$ so~$\W^{g}(i)=\{i\}$ and for any complex~$K'$,~$g_*(K')=K'$.

The language~$g(L)\subseteq\{0,1\}^I$ is the set of functions that are constant on~$S$ and constantly~$0$ outside~$S$. By Example \ref{ex_constant}, a complex generates~$g(L)$ if and only if it contains~$S$. If a complex~$K'$ generates~$L$, then~$g_*(K')$ generates~$g(L)$ so~$g_*(K')$ contains~$S$. As~$g_*(K')=K'$,~$K'$ contains~$S$. As this is true for any~$S\in K$,~$K'$ contains~$K$.
\end{proof}

\subsection{The join}
A simple strategy to generate a language~$L\subseteq A^I$ is to first generate a projection of~$L$ on a set~$J\subseteq I$, and then let the remaining elements in~$J\setminus I$ extend the result. The complex generating~$L$ using this strategy can be simply expressed in terms of the complex generating its projection, using the join operator.

Let~$K_0,K_1$ be two complexes on disjoint vertex sets~$I_0,I_1$ respectively. Their \textbf{join}~$K_0\star K_1$ is the complex over~$I=I_0\sqcup I_1$ whose simplices are~$S_0\sqcup S_1$, with~$S_0\in K_0$ and~$S_1\in K_1$ (see \cite{Kozlov07}).

\begin{example}
If~$K$ is a complex over~$I$ and~$a\notin I$, then~$\{a\}\star K$ is the complex over~$\{a\}\sqcup I$ whose maximal simplices are obtained by adding~$a$ to the maximal simplices of~$K$ (it is usually called the \emph{cone} of~$K$).
\end{example}

\begin{example}
Let~$K$ be a complex over~$I$ and let~$a,b\notin I$ be distinct. Let~$\{\{a\},\{b\}\}$ be the complex made of two vertices and no edge. The complex~$\{\{a\},\{b\}\}\star K$ is the complex whose maximal simplices are~$\{a\}\sqcup S$ and~$\{b\}\sqcup S$, for every maximal simplex~$S$ of~$K$ (it is usually called the \emph{suspension} of~$K$).
\end{example}

We denote by~$\Delta_I$ the full complex with vertices in~$I$.

\begin{proposition}\label{prop_join}
Let~$L\subseteq A^I$ and~$J\subseteq I$. If a complex~$K$ generates~$\pi_J(L)$, then~$\Delta_{I\setminus J}\star K$ generates~$L$.
\end{proposition}
\begin{proof}
Intuitively, the output cells in~$J$ keep the same generation procedure to produce an element~$y$ of~$\pi_J(L)$, and the other output cells have access to the entire input and collectively choose the extension of~$y$ to be produced.

More precisely, let~$f:C^H\to A^J$ be a function generating~$\pi_J(L)$, with~$K_f\subseteq K$. We can assume that~$|C|\geq |L|$. We pick a function~$g:C\times C^H\to A^I$ such that for each~$x\in C^H$,~$\{g(c,x):c\in C\}$ is the set of elements of~$L$ extending~$f(x)$. We can reformulate the input space of~$g$ as~$C^{\{a\}\sqcup H}$ with~$a\notin H$.

We show that~$K_g\subseteq \Delta_{I\setminus J}\star K$, which means that for every simplex~$S\in K_g$,~$S\cap J\in K$. Observe that an output cell~$i\in J$ does not need access to the value of~$a$, because~$g(c,x)_i=f(x)_i$ does not depend on~$c$. In other words, one has~$\W_g(i)=\W_f(i)$ for~$i\in J$. If~$S\in K_g$, then~$\bigcap_{i\in S\cap J}\W_f(i)=\bigcap_{i\in S}\W_g(i)\neq \emptyset$, therefore~$S\cap J\in K$, i.e.~$S\in \Delta_{I\setminus J}\star K$.
\end{proof}

We state two particular cases.
\begin{corollary}\label{cor_star_singleton}
Let~$L\subseteq A^I$,~$i\in I$ and~$J=I\setminus \{i\}$. If~$K$ generates~$\pi_{J}(L)$, then~$\{i\}\star \{K\}$ generates~$L$.
\end{corollary}
\begin{proof}
By Proposition \ref{prop_join},~$\Delta_{I\setminus J}\star K$ generates~$L$, and~$\Delta_{I\setminus J}$ is~$\{i\}$.
\end{proof}

Given a language~$L\subseteq A^I$, we say that two cells~$i,j\in I$ are \textbf{independent} w.r.t.~$L$ if~$\pi_{\{i,j\}}(L)=\pi_{\{i\}}(L)\times \pi_{\{j\}}(L)$, i.e.~if the contents of~$i$ and~$j$ can be chosen independently of each other. In other words, it means that~$\pi_{\{i,j\}}(L)$ is not irreducible. Intuitively, if~$i,j$ are independent, then they do not need to communicate directly in a generation procedure and indeed, it is reflected by the existence of a complex generating~$L$ that does not contain the edge~$\{i,j\}$.

\begin{corollary}\label{cor_two_indep}
Let~$i,j\in I$ be distinct. The two cells~$i,j$ are independent w.r.t.~$L$ if and only if~$L$ is generated by~$\{I\setminus \{i\},I\setminus\{j\}\}$, which is the complex made of all the simplices that do not contain~$\{i,j\}$.
\end{corollary}

\begin{proof}
Let~$J=\{i,j\}$. If~$i,j$ are independent, then the complex~$\{\{i\},\{j\}\}$ generates~$\pi_J(L)$, so~$\Delta_{I\setminus \{i,j\}}\star \{\{i\},\{j\}\}$ generates~$L$ by Proposition \ref{prop_join}. The maximal simplices of that complex are~$I\setminus \{i\}$ and~$I\setminus\{j\}$.

Conversely, if that complex generates~$L$ then its restriction to~$\{i,j\}$ generates~$\pi_{\{i,j\}}(L)$. The restriction of the complex is~$\{\{i\},\{j\}\}$, which is disconnected, so~$\pi_{\{i,j\}}(L)=\pi_{\{i\}}(L)\times \pi_{\{j\}}(L)$ (Proposition \ref{prop_irreducible}).
\end{proof}

However, in that case, it may happen that some minimal complex generating~$L$ contains the edge~$\{i,j\}$ (see for instance the language of even sequences in Section \ref{sec_even} below).

\section{Simple languages}\label{sec_simple_examples}
We now examine many particular languages, which are all binary, i.e.~subsets of~$\{0,1\}^n$ for some~$n\in\N$.

\subsection{Even number of \texorpdfstring{$1$}{1}'s}\label{sec_even}
Let~$\Ev_I\subseteq\{0,1\}^I$ be the set of strings containing an even number of~$1$'s (we call them even strings). Note that~$\Ev_I$ is invariant under the action of the whole group of permutations of~$I$, so the family of complexes generating~$\Ev_I$ is closed under permutations of the vertices.
 
\begin{proposition}
A complex generates~$\Ev_I$ iff it is connected. Therefore, the minimal complexes generating~$\Ev_I$ are the trees spanning~$I$.
\end{proposition}
\begin{proof}
Note that~$\Ev_I$ is irreducible: for any~$J\subsetneq I$,~$\pi_J(\Ev_I)=\{0,1\}^J$ so for any non-trivial partition~$I=I_0\sqcup I_1$,~$\pi_{I_0}(\Ev_I)\times \pi_{I_1}(\Ev_I)=\{0,1\}^I\neq \Ev_I$. Therefore, every complex generating~$\Ev_I$ must be connected by Proposition \ref{prop_irreducible}.

Conversely, every connected complex contains a spanning tree, so it is sufficient to show that every spanning tree generates~$\Ev_I$. Let~$T$ be such a tree and consider the following local rule upon a binary assignment to its edges: each vertex is assigned the sum modulo~$2$ of the weights of its incident edges. This procedure indeed generates even strings, because the sum of the weights of the vertices is twice the sum of the weights of the edges. It produces every even string~$x$: choose an edge incident to a leave, give it the only possible weight, remove that edge and iterate this process; this way we can freely choose the output values of all the vertices except the last one, but its value is the only one that makes the string even (we are essentially applying the Gaussian elimination algorithm to show that~$\Ev_I$ is the image of the linear map between vector spaces over the field~$\Z/2\Z$, represented by the incident matrix of the tree).
\end{proof}

The case of an odd number of~$1$'s is similar, because it is obtained by flipping one bit.  More precisely, let~$\Od_I$ be the set of strings having an \emph{odd} number of~$1$'s. Fix any~$i_0\in I$ and let~$f:\{0,1\}^I\to\{0,1\}^I$ be defined by~$f(x)_{i_0}=1-x_{i_0}$ and~$f(x)_i=x_i$ for~$i\neq i_0$. One has~$f(\Ev_I)=\Od_I$ and~$f(\Od_I)=\Ev_I$, and for any complex~$K$,~$f_*(K)=K$ as~$\W^f(i)=\{i\}$ for all~$i\in I$. Therefore, the same complexes generate~$\Ev_I$ and~$\Od_I$.

\subsection{Non-decreasing sequences}\label{sec_non_dec}
Let~$\ND_n\subseteq \{0,1\}^n$ be the set of non-decreasing binary sequences over~$I_n$, i.e.~$\ND_n=\{0^k1^{n-1-k}:0\leq k<n\}$.
\begin{proposition}\label{prop_non_decr}
The only complex realizing~$\ND_n$ is~$\Delta_{I_n}$, the full complex over~$I_n$.
\end{proposition}

In order to prove the result, let us introduce some simple notions. Fix some function~$f:B^J\to A^I$. 

\begin{definition}\label{def_partial_input}
A \textbf{partial input} is~$\alpha\in B^D$ for some~$D\subseteq J$. If~$E\subseteq J$, then~$\restr{\alpha}{E}\in B^{D\cap E}$ is the restriction of~$\alpha$ to~$D\cap E$. For~$i\in I$, we write~$f_i(\alpha)=1$ if for all~$x\in B^J$ extending~$\alpha$, one has~$f_i(x)=1$.
\end{definition}

\begin{proof}
Let~$f:B^J\to\{0,1\}^n$ generate~$\ND_n$. Let~$A_0=I_n$ and for~$0\leq i<n$, let~$A_{i+1}=A_{i}\cap\W_f(i)$. Our goal is to show that~$A_n=\up[f]{I_n}$ is non-empty, i.e.~that~$K_f$ contains the full simplex~$I_n$.

We make two observations:
\begin{enumerate}[(i)]
\item If~$f_i(\alpha)=1$, then~$f_i(\restr{\alpha}{\W_f(i)})=1$ because~$f_i(x)$ is determined by~$\restr{x}{\W_f(i)}$,
\item If~$f_i(\alpha)=1$, then~$f_{i+1}(\alpha)=1$ as~$f$ only produces non-decreasing sequences.
\end{enumerate}

Therefore, we can iteratively derive the following chain of implications, starting from  some~$x$ such that~$f(x)=1^n$:
\[
\begin{array}{cccc}
& f_0(x)=1&\stackrel{(i)}{\implies}&f_0(x\res{A_1})=1\\
\stackrel{(ii)}{\implies}& f_1(\restr{x}{A_1})=1&\stackrel{(i)}{\implies}& f_1(\restr{x}{A_2})=1\\
\stackrel{(ii)}{\implies} & \ldots&&\\
\stackrel{(ii)}{\implies} & f_{n-1}(\restr{x}{A_{n-1}})=1&\stackrel{(i)}{\implies}& f_{n-1}(\restr{x}{A_{n}})=1.
\end{array}
\]
If~$A_n$ was empty, then the last statement would mean that~$f_{n-1}$ is constantly~$1$, because every input would extend~$\restr{x}{A_n}$ which nowhere defined. But~$f_{n-1}$ also takes value~$0$ as~$0^n\in L_n$. Therefore,~$A_n\neq\emptyset$, which is what we wanted to prove.
%
\end{proof}

More generally, this argument shows that if~$L\subseteq \{0,\ldots,k\}^n$ contains the constant sequences~$0^n$ and~$k^n$ and is contained in the set of non-decreasing sequences, then the only complex generating~$L$ is the full complex over~$I_n$.

\subsection{Non-constant sequences}
Assume that~$|A|\geq 2$ and let~$\NC_I\subseteq A^I$ be the set of non-constant sequences.
\begin{proposition}
A complex generates~$\NC_I$ iff it is connected. In other words, the minimal complexes generating~$\NC_I$ are the trees spanning~$I$.
\end{proposition}

\begin{proof}
By Proposition \ref{prop_irreducible}, a generating complex must be connected because~$\NC_I$ is irreducible: if~$I=I_0\sqcup I_1$ is a non-trivial partition, then~$\pi_{I_0}(\NC_I)\times \pi_{I_1}(\NC_I)=A^I\neq \NC_I$.

Conversely, if a complex is connected then it contains a spanning tree, so it is sufficient to prove that every spanning tree~$T$ generates~$\NC_I$. We give a simple procedure whose input alphabet is~$\NC_I$. A more economical input alphabet would be possible, with a slightly more complicated procedure.

We fix a function~$h:A\to A$ such that~$h(a)\neq a$ for all~$a\in A$. We fix a root, and each vertex which is not a leaf chooses a distinguished children. Each edge is assigned an element of~$\NC_I$. The rule of a vertex~$i$ is as follows: if all its incident edges have the same value~$x$, then~$i$ takes value~$x_i$; otherwise it takes value~$h(x_j)$, where~$j$ is its distinguished children and~$x$ is the value assigned to the edge~$(i,j)$ (note that the second case cannot happen for a leaf, which has degree one hence sees the same value on its incident edge(s)).

Every~$x\in \NC_I$ can be obtained by assigning~$x$ to all the edges. Conversely, we fix some input ans show that it produces a non-constant sequence. If the input gives the same value~$x$ to all the edges, then the output is~$x\in\NC_I$. Otherwise let~$E$ be the set of vertices having at least two incident edges with distinct values. There exists~$i\in E$ such that its distinguished children~$j$ is outside~$E$ (indeed, start from some~$i\in E$ and follow its distinguished children as long as it belongs to~$E$: as the leaves are outside~$E$, we eventually find such an~$i$). As~$j$ sees the same value~$x$ on its incident edges, it takes value~$x_j$. As~$i$ does not, it takes value~$h(x_j)\neq x_j$. Therefore, the output sequence is non-constant.
\end{proof}

\subsection{Upwards closed languages}
We can identify~$\{0,1\}^I$ with~$2^I$, the powerset of~$I$, seeing a binary sequence as the characteristic function of a subset of~$I$. We endow~$2^I$ with the inclusion ordering and say that~$L\subseteq 2^I$ is \textbf{upwards closed} if for all~$U,V\in 2^I$, if~$U\in L$ and all~$U\subseteq V$ then~$V\in L$.

For any upwards closed language~$L$, we obtain a characterization of the complexes generating~$L$. It turns out that the minimal complexes generating such languages are graphs, which means that only pairwise communication is needed between vertices.

\begin{definition}
Let~$L\subseteq 2^I$ be upwards closed. Say that a graph~$G$ on vertex set~$I$ is \textbf{$L$-connected} if for every maximal set~$W\notin L$, removing~$W$ from~$G$ results in a connected graph.
\end{definition}

\begin{example}[$k$-connected graphs]
The classical class of~$k$-connected graphs \cite{Diestel17} is an instance of this notion. Let~$k\in\N$,~$|I|\geq k+1$ and let~$\card_{\geq k}\subseteq 2^I$ be the collection of sets of size at least~$k$. It is upwards closed and the maximal elements~$W\notin L$ are the sets of size~$k-1$. A graph is~$\card_{\geq k}$-connected if and only if it is~$k$-connected.

The notion of~$k$-connected graph usually assumes that~$|I|\geq k+1$. When~$|I|=k$, removing~$k-1$ vertices results in a single vertex, which is a connected graph, so every graph on~$k$ vertices is~$\card_{\geq k}$-connected.
\end{example}

\begin{remark}
For certain~$L$, being~$L$-connected implies that for every~$W\notin L$, removing~$W$ yields a connected graph, but it is not always the case.

For instance, for~$I=\{0,1,2\}$ and~$L$ be the collection of sets containing~$0$ or~$1$. The complement of~$L$ is~$\{\emptyset,\{2\}\}$. Let~$G$ be the graph on~$I$ having only one edge between~$0$ and~$1$. It is~$L$-connected but not connected, i.e.~removing~$\emptyset$ from~$G$ does not yield a connected graph.
\end{remark}

The \textbf{$1$-skeleton} of a complex is the set of edges belonging to that complex, it is therefore a graph.
\begin{theorem}[Generating upwards closed languages]\label{thm_L_connected}
Let~$L\subseteq 2^V$ be upwards closed. A complex on~$V$ generates~$L$ if and only if its~$1$-skeleton is~$L$-connected. In other words, the minimal complexes generating~$L$ are the~$L$-connected graphs.
\end{theorem}
\begin{proof}
Let~$G$ be an~$L$-connected graph. To each edge we assign an element of~$L$. If all the edges that are incident to a vertex~$i$ are assigned the same element~$U\in L$, then~$i$ takes value~$U(i)$, otherwise it takes value~$1$. We show that this function generates~$L$. First, every element of~$L$ is reached, by assigning it to every edge. Now assume for a contradiction that the function outputs some~$U\notin L$, and let~$x$ be a corresponding input. Let~$W\supseteq U$ be a maximal element outside~$L$. For each~$i\notin W$,~$U(i)=0$ so all the edges that are incident to~$i$ have the same value~$V_i$ in~$x$, and~$V_i(i)=0$. As~$I\setminus W$ is connected by assumption on the graph,~$V:=V_i$ does not depend on~$i\notin W$. One has~$V\in L$ and~$V(i)=0$ for all~$i\notin W$, i.e.~$V\subseteq W$, so~$W\in L$ which is a contradiction. Therefore, the function only produces elements of~$L$, and produces all of them, so~$G$ generates~$L$. If the~$1$-skeleton~$G$ of a complex~$K$ is~$L$-connected, then~$G$ generates~$L$ so~$K\supseteq G$ generates~$L$ as well.

Conversely, let~$G$ be the~$1$-skeleton of a complex~$K$ generating~$L$ and assume for a contradiction that for some maximal set~$W\notin L$,~$I\setminus W$ is disconnected in~$G$. Let~$I\setminus W=V_0\sqcup V_1$ be a non-trivial disconnection. As~$W$ is maximal outside~$L$ and~$V_0,V_1$ are non-empty, one has~$W_0:=W\sqcup V_0\in L$ and~$W_1:=W\sqcup V_1\in L$. Let~$x_0,x_1$ be corresponding inputs on~$K$. We define~$x$ that coincides with~$x_0$ on the simplices that contain a vertex in~$V_1$, with~$x_1$ on the simplices that contain a vertex in~$V_0$, and that is arbitrary elsewhere. Note that it is well-defined because there is no edge, hence no simplex, containing vertices in both~$V_0$ and~$V_1$. Let~$U$ be the output of the generating function on~$x$. For each vertex~$i\in V_0$,~$x$ coincides with~$x_1$ on its incident simplices, so~$U(i)=W_1(i)=0$. Similarly, for each vertex~$i\in V_1$,~$U(i)=W_0(i)=0$. Therefore,~$U(i)=0$ for all~$i\in I\setminus W$, i.e.~$U\subseteq W$. As~$U$ is produced by the function, it belongs to~$L$ which is upwards closed, so~$W\in L$, contradicting the choice of~$W$. Therefore~$I\setminus W$ is connected in~$G$, and finally~$G$ is~$L$-connected.
\end{proof}

We state the corresponding formulation of the result for downwards closed languages.
\begin{corollary}[Generating downwards closed languages]\label{cor_downwards}
Let~$L\subseteq 2^I$ be downwards closed. A complex generates~$L$ if and only if its~$1$-skeleton~$G$ satisfies the property that for every minimal~$W\notin L$, the induced subgraph of~$G$ on~$W$ is connected.
\end{corollary}
\begin{proof}
It can be proved similarly, or derived from the previous result. The complementation function~$c:2^I\to 2^I$ induces a bijection between~$L$ and~$c(L)$ which is upwards closed. Moreover, its evaluation is local in the sense that~$\W^c(i)=\{i\}$ for all~$i\in I$. As a result,~$c_*(K)=K$ for any complex~$K$ so the same complexes generate a language~$L$ and its image~$c(L)$ (note that~$c$ is its own inverse, so it sends~$c(L)$ to~$L$). The property in the statement precisely means that~$c(L)$ is~$c(L)$-connected.
\end{proof}
\subsubsection{Examples}\label{sec_card}
We apply Theorem \ref{thm_L_connected} to simple examples.

\begin{corollary}[Cardinality at least $k$]\label{cor_at_least}
Let~$k\in\N$,~$|I|\geq k+1$ and~$\card_{\geq k}\subseteq 2^I$ be the collection of subsets of~$I$ of size at least~$k$. A complex generates~$\card_{\geq k}$ if and only if its~$1$-skeleton is a~$k$-connected graph.
\end{corollary}
If~$|I|=k$, then~$\card_{\geq k}$ is a singleton and is generated by any complex over~$I$ containing the singletons. It agrees with the fact that every graph on~$k$-vertices is~$\card_{\geq k}$-connected. In particular, a complex generates~$\card_{\geq 1}$ if and only if it is connected, so the minimal ones are the spanning trees.

Similarly, we can define~$\card_{\leq k}\subseteq 2^V$, which is the collection of subsets of size at most~$k$.  
\begin{corollary}[Cardinality at most $k$]\label{cor_at_most}
Let~$k\in\N$ and~$|I|\geq k+1$. A complex generates~$\card_{\leq k}$ if and only if its~$1$-skeleton is a graph whose induced subgraphs of size~$k+1$ are all connected.
\end{corollary}

In particular, the only minimal complex generating~$\card_{\leq 1}$ is the complete graph.

\begin{remark}[At most one occurrence of~$1$]\label{rmk_at_most1}
 We give a more concrete way of generating~$\card_{\leq 1}$ by the complete graph. On the input side, each edge of the graph is assigned a bit, which is interpreted as a choice between its two endpoints (one may for instance orient the edges and interpret~$0$ as the source and~$1$ as the target of the edge). On the output side, a vertex is assigned~$1$ if and only if it has been chosen by all its incident edges. As the graph is the complete graph, at most one vertex can take value~$1$, and it is easy to see that every sequence with at most one occurrence of~$1$ can be obtained this way.
\end{remark}

We apply the result to another class of downwards closed languages. Let~$G=(V,E)$ be a simple graph and let~$L_G\subseteq \{0,1\}^V$ be the set of binary colorings~$(x_v)_{v\in V}$ of~$V$ such that there is no edge~$\{u,v\}\in E$ satisfying~$x_u=x_v=1$.
\begin{corollary}
The only minimal complex generating~$L_G$ is~$G$.
\end{corollary}
\begin{proof}
$L_G$ is downwards closed so we can apply Corollary \ref{cor_downwards}. The minimal colorings outside~$L_G$ are the ones that assign value~$1$ to the two ends of an edge and~$0$ to the other vertices. Therefore, a complex~$K$ generates~$L_G$ if and only if every edge of~$E$ belongs to~$K$, i.e.~iff~$K$ contains~$G$.
\end{proof}

\subsection{One or \texorpdfstring{$n$ occurrences of $1$}{n occurrences of 1}}
Let
\begin{equation*}
L=\{u\in \{0,1\}^n:|u|_1=1\text{ or }n\}.
\end{equation*}
For~$n=2$, only the line segment realizes~$L$, because the two vertices are not independent (they can both take value~$0$, but not at the same time).

\begin{proposition}
Let~$n\geq 3$. The minimal complexes generating~$L$ are~$\{I_n\setminus \{a\},I_n\setminus \{b\}\}$ for distinct~$a,b\in I_n$.
\end{proposition}
\begin{proof}
For~$n\geq 3$, for any pair~$\{a,b\}$ of distinct vertices,~$a,b$ are independent so~$\{I_n\setminus \{a\},I_n\setminus \{b\}\}$ generates~$L$ by Corollary \ref{cor_two_indep}.

Conversely, let~$K$ be a complex generating~$L$. We show that it contains~$\{I_n\setminus \{a\},I_n\setminus \{b\}\}$ for some distinct~$a,b$.

\begin{claim}
For any~$i\in I_n$,~$K$ contains a simplex~$I_n\setminus \{a\}$, for some~$a\neq i$. \end{claim}
\begin{proof}
Let~$i\in I_n$. Let~$\F$ be the set of partial inputs~$\alpha$ such that~$f_j(\alpha)=1$ for all~$j$. If~$\alpha\in \F$ and~$j,k\in I_n$ are distinct, then~$\beta:=\restr{\alpha}{\W_f(j)\cup\W_f(k)}\in\F$. Indeed,~$f_j(\alpha)=1$ and~$f_k(\alpha)=1$ implies that~$f_j(\restr{\alpha}{\W_f(j)}=1$ and~$f_k(\restr{\alpha}{\W_f(k)}=1$, which implies that~$f_j(\beta)=1$ and~$f_k(\beta)=1$, which implies that~$f_l(\beta)=1$ for all~$l$, i.e.~$\beta\in \F$.

Let~$x$ be an input such that~$f(x)=1^n$. One has~$x\in\F$ and by iterating the previous observation, the restriction~$\alpha$ of~$x$ to~$\bigcap_{j\neq k}\W_f(j)\cup\W_f(k)$ belongs to~$\F$. Therefore,~$f_i(\alpha)=1$ hence~$f_i(\restr{\alpha}{\W_f(i)})=1$.

As~$f_i$ can take value~$0$,~$\W_f(i)\cap\bigcap_{j\neq k}\W_f(j)\cup\W_f(k)\neq \emptyset$. Let~$h$ belong to that set. Its dual window~$\W^f(h)$ contains~$i$ and intersects every pair~$\{j,k\}$ with~$j\neq k$. Therefore, there is at most one element~$a\in I_n$ outside~$\W^f(h)$, and that element cannot be~$i$. We conclude that~$I_n\setminus\{a\}$ is contained in~$\W^f(h)$ hence belongs to~$K_f$.
\end{proof}

Pick some arbitrary~$i\in I_n$. By the previous claim,~$K$ contains a simplex~$I_n\setminus \{a\}$ for some~$a$. Applying again the claim to~$i=a$,~$K$ contains a simplex~$I_n\setminus \{b\}$ with~$b\neq a$. Therefore,~$K$ contains~$(I_n\setminus \{a\},I_n\setminus \{b\})$ for some~$a\neq b$.
\end{proof}

The same result holds for~$Y=\{u\in\{0,1\}^n:|u|_1=0\text{ or }1\text{ or }n\}$. Indeed, let~$w_a$ take an extra value~$\bot$, which allows for all the vertices to take value~$0$. The rest of the argument applies equally to~$Y$ because it has the same property that if~$u\in Y$ and~$|u|_1\geq 2$, then~$|u|_1=n$.
%

\section{Partial results on other languages}\label{sec_complicated_examples}
We present two languages for which the analysis turns out to be more difficult.
\subsection{One occurrence of \texorpdfstring{$1$}{1}}\label{sec_unique}
\begin{definition}
Let~$\U_n\subseteq \{0,1\}^{n}$ be the set of sequences containing a unique occurrence of~$1$.
\end{definition}
We identify a family of minimal complexes but so far, we do not know whether they are the only ones.

For~$n=1$, the empty complex generates~$\U_1=\{1\}$. For~$n=2$,~$\U_2=\{01,10\}$ is irreducible so the only complex generating~$\U_2$ is the full complex over~$I_2$, i.e.~the edge. We now assume that~$n\geq 3$.

We first give a necessary condition for a complex to generate~$\U_n$: every pair of vertices must belong to a solid triangle.
\begin{proposition}\label{prop_unique_triangle}
Let~$n\geq 3$. If~$K$ generates~$\U_n$, then for every pair of distinct elements~$a,b\in I_n$, there exists~$c\in I_n\setminus \{a,b\}$ such that~$\{a,b,c\}\in K$.
\end{proposition}
\begin{proof}
We first show the result for~$n=3$, namely that the only complex generating~$\U_3$ is the triangle. Let~$F$ generate~$\U_3$. For each~$i\in \{0,1,2\}$, there exists an input~$x_i$ such that~$F_i(x_i)=1$. We can derive the following constraints:
\begin{align*}
F_i(x_i)=1&\implies F_i(\restr{x_i}{\W_i})=1\\
&\implies F_{i+1}(\restr{x_i}{\W_i\cap \W_{i+1}})=0,
\end{align*}
where~$i+1$ should be taken modulo~$3$. Let~$\alpha_i=\restr{x_i}{\W_i\cap\W_{i+1}}$. If~$K_F$ does not contain the triangle~$\{0,1,2\}$, then~$\W_0\cap\W_1\cap\W_2=\emptyset$ so the~$\alpha_i$'s have disjoint domains hence are compatible; therefore,~$\alpha=\alpha_0\sqcup\alpha_1\sqcup\alpha_2$ is well-defined and satisfies~$F(\alpha)=000$. As this sequence does not belong to~$\U_3$, we obtain a contradiction.

We now prove the general case. Let~$K$ be a complex generating~$\U_n$, and let~$a,b\in I_n$ be distinct. Let~$f:\U_n\to\U_3$ be defined by
\begin{align*}
f_0(x)&=x_a,\\
f_1(x)&=x_b,\\
f_2(x)&=\max\{x_i:i\in I_n\setminus \{a,b\}\}.
\end{align*}
One has~$f(\U_n)=\U_3$ and~$\W^f(a)=\{0\}$,~$\W^f(b)=\{1\}$ and~$\W^f(c)=\{2\}$ for all~$c\in I_n\setminus\{a,b\}$. Therefore,~$f_*(K)$ generates~$\U_3$ so~$f_*(K)$ contains~$\{0,1,2\}$ by the first part of the proof. It means that there exists a simplex~$S\in K$ containing~$a,b$ and some~$c\in I_n\setminus \{a,b\}$.
\end{proof}

We now identify a family of minimal complexes generating~$\U_n$.

Let~$n\geq 3$ and~$a\in I_n$. Let~$C_{I_n\setminus \{a\}}$ be the complete graph on~$I_n\setminus \{a\}$ and let~$K_a:=\{a\}\star C_{I_n\setminus\{a\}}$ be made of all the triangles containing~$a$.

\begin{proposition}
The complex~$K_a$ is minimal generating~$\U_n$.
\end{proposition}
\begin{proof}
Observe that~$\pi_{I_n\setminus \{a\}}(\U_n)=\card_{\leq 1}$. As~$C_{I_n\setminus\{a\}}$ generates~$\card_{\leq 1}$ by Corollary \ref{cor_at_most},~$K_a=\{a\}\star K_{I_n\setminus \{a\}}$ generates~$\U_n$ by Corollary \ref{cor_star_singleton}.

Let~$K\subseteq K_a$ generate~$\U_n$. Let~$b,c\in I_n\setminus \{a\}$ be distinct. By Proposition \ref{prop_unique_triangle}, there exists~$d\in I_n\setminus \{b,c\}$ such that~$K$ contains~$\{b,c,d\}$. The only triangle in~$K_a$ containing~$b$ and~$c$ is~$\{a,b,c\}$, so~$d=a$ and~$K$ contains~$\{a,b,c\}$. Therefore,~$K$ contains every triangle of~$K_a$, so~$K=K_a$.
\end{proof}

\paragraph{Small values of \texorpdfstring{$n$}{n}.}
For small~$n$, we can show that the complexes~$K_a$ are the only minimal complexes generating~$\U_n$.
\begin{proposition}\label{prop_unique_5}
For~$3\leq n\leq 5$, the only minimal complexes generating~$\U_n$ are of the form~$K_a$.
\end{proposition}
\begin{proof}
For~$n=3$, we know that the only complex generating~$\U_3$ is the full complex, which is~$K_a$ for any~$a$.

For~$n=4$, let~$K$ be minimal generating~$\U_4$. $K$ cannot contain every triangle, otherwise it properly contains any~$K_a$ hence it is not minimal. Assume w.l.o.g.~that~$K$ does not contains the triangle~$abc$. As every edge belongs to a triangle,~$K$ contains the triangles~$abd,bcd,acd$, i.e.~$K$ contains~$K_d$ hence~$K=K_d$.

Let~$n=5$. We first show the following lemma which, combined with Proposition \ref{prop_unique_triangle}, will imply the result.

\begin{lemma}\label{lem_abc_ade}
If~$K$ generates~$\U_5$, then~$K$ contains the triangle~$abc$ or the triangle~$ade$.
\end{lemma}
\begin{proof}
Assume for a contradiction that it is not the case, and let~$f:B^J\to \{0,1\}^5$ generate~$\U_5$ with~$K_f\subseteq K$. Let~$\alpha$ be an arbitrary partial input whose domain is the complement of~$\W_f(a)$.
\begin{claim}
One must have~$f_b(\alpha)=0$ or~$f_c(\alpha)=0$. Similarly,~$f_d(\alpha)=0$ or~$f_e(\alpha)=0$.
\end{claim}
\begin{proof}
The pairs~$(b,c)$ and~$(d,e)$ play a symmetric role, so we only need to prove the first. part. The idea is that once filling the input cells with~$\alpha$,~$b$ and~$c$ cannot communicate any more, because~$\W_f(b)\cap\W_f(c)$ is disjoint from~$\W_f(a)$. It each one of them can take value~$1$, then they can take that value at the same time, producing an invalid output.

More formally, if~$x,y$ are inputs extending~$\alpha$, such that~$f_b(x)=1$ and~$f_c(y)=1$, then we let~$z$ be any input that coincides with~$x$ on~$\W_f(b)$ and with~$y$ on~$\W_f(c)$ (which is possible because~$x$ and~$y$ agree on~$\W_f(b)\cap\W_f(c)\subseteq\dom(\alpha)$). One has~$f_b(z)=f_c(z)=1$, so~$f(z)\notin \U_5$.
\end{proof}
By symmetry, we can assume that~$f_c(\alpha)=0$ and~$f_e(\alpha)=0$. Let~$y$ be such that~$f_a(y)=1$ and let
\[
\beta=\alpha\sqcup \restr{y}{\up[f]{\{a,b,d\}}}.
\]
\begin{claim}

One must have~$f_b(\beta)=0$ or~$f_d(\beta)=0$.
\end{claim}
\begin{proof}
The idea is the same: no communication is possible between~$b$ and~$d$ outside~$\dom(\beta)$, because~$\W_f(b)\cap\W_f(d)\setminus \dom(\alpha)=\up[f]{\{a,b,d\}}$. Therefore, one of them must constantly take value~$0$ on extensions of~$\beta$.
\end{proof}
By symmetry, we can assume that~$f_d(\beta)=0$. As~$f_a(y)=1$, one has~$f_b(\restr{y}{\up{\{a,b\}}})=0$. Therefore, if~$\gamma=\alpha\sqcup \restr{y}{\up[f]{\{a,b\}}}$, which extends~$\beta$, then~$f_b(\gamma)=0$.

Finally, let~$x$ be such that~$f_c(x)=1$ and let~$\delta=\restr{x}{\up[f]{\{c,a\}}}$. One has~$f_a(\delta)=0$. Observe that~$\gamma$ and~$\delta$ have disjoint domains because~$abc$ is missing in~$K$, so they are compatible. But~$f(\gamma\sqcup\delta)=00000\notin \U_5$, which is a contradiction.
\end{proof}
Of course, all the symmetric versions of Lemma \ref{lem_abc_ade} hold. It implies that if the triangle~$abc$ is missing in~$K$, then all the triangles containing~$d$ and~$e$ are in~$K$, and all the symmetric versions of this property.

Now assume that~$K$ does not contain any~$K_v$. It means that for each vertex~$v$, some triangle containing~$v$ is missing in~$K$.

Start with a triangle that is missing in~$K$ and assume by symmetry that it is~$abc$. By Lemma \ref{lem_abc_ade},~$K$ contains all the triangles containing~$d$ and~$e$.

By assumption, some triangle containing~$d$ is missing in~$K$. It cannot contain~$e$, so it is~$abd,acd$ or~$bcd$. By symmetry, we can assume that~$abd$ is missing. Lemma \ref{lem_abc_ade} implies that~$K$ contains all the triangles containing~$c$ and~$e$.

Finally, some triangle containing~$e$ is missing, it cannot contain~$d$ or~$c$ so it is~$abe$. As a result, all the triangles containing~$a$ and~$b$ are missing, contradicting Proposition \ref{prop_unique_triangle}.

Therefore,~$K$ contains some~$K_v$, so~$K=K_v$ by minimality.
\end{proof}

We leave several questions open:
\begin{itemize}
\item For~$n\geq 6$, are there other minimal complexes generating~$\U_n$?
\item If the answer is positive, are they only made of triangles?
\item Are the~$K_a$'s the only minimal complexes generating~$\U_n$ and containing only triangles?
\end{itemize}


\subsection{Identical consecutive letters}\label{sec_consecutive}

Let~$A$ be a finite alphabet and~$\Eq_n(A)\subseteq A^n$ be the set of strings having at least two identical consecutive letters:~$w\in \Eq_n(A)\iff\exists i,w_i=w_{i+1}$.

The class of complexes generating~$\Eq_n(A)$ depends on the size of the alphabet.

\subsubsection{Two letters}
The binary alphabet~$A=\{0,1\}$ is a very special case because the local condition of having two identical consecutive letters can be forced by choosing two distant letters: for instance with~$n=100$, if~$w_0=w_{99}=0$ then~$w\in \Eq_n(A)$.

\begin{proposition}\label{prop_two_letters}
If~$|A|=2$ and~$n\geq 3$, then the minimal complexes generating~$\Eq_n(A)$ are the spanning trees.
\end{proposition}
\begin{proof}
Assume that~$A=\{0,1\}$. $\Eq_n(A)$ is irreducible so every complex generating~$\Eq_n(A)$ must be connected, hence contain a spanning tree. Conversely, observe that~$w\in \Eq_n(A)\iff \exists i,j,w_i+w_j+i+j=1\bmod 2$. Indeed, if~$w\notin \Eq_n(A)$ then~$\forall i,w_i=i\bmod 2$ or~$\forall i,w_i=i+1\bmod 2$, and in both case~$w_i+w_j+i+j=0\bmod 2$ for all~$i,j$.

Let~$T$ be a rooted spanning tree and consider the following procedure. Each edge is labeled with an element of~$\Eq_n(A)$. The rule of a vertex~$i$ is as follows: if all its incident edges have the same label~$w$, then~$i$ takes value~$w_i$; otherwise let~$w$ be the label on the edge between~$i$ and its leftmost child~$j$, and let~$i$ take value~$w_j+i+j+1\bmod 2$.

Every element~$w$ of~$\Eq_n(A)$ is reached, by assigning~$w$ to all the edges. Conversely, if all the edges have the same value~$w$, then the output is~$w$ which belongs to~$\Eq_n(A)$. Otherwise, let~$i$ be a deepest vertex whose incident edges do not have the same value (deepest means that its descendants do not have this property). Observe that~$i$ is not a leaf. As its leftmost child~$j$ sees the same value~$w$ on its incident edges, it takes value~$w_j$. As~$i$ does not see the same values on its incident edges, it takes value~$w_j+i+j+1\bmod 2$. As a result, the output~$x$ satisfies~$x_i+x_j+i+j=2(w_j+i+j)+1=1\bmod 2$, so~$x\in \Eq_n(A)$.
\end{proof}

\subsubsection{More letters}

The proof of Proposition \ref{prop_two_letters} does not work with more than two letters, and indeed we show that not every spanning tree generates~$\Eq_n(A)$. We give partial results, leaving open the problem of finding a complete characterization of the complexes generating~$\Eq_n(A)$. We conjecture that all the minimal complexes generating~$\Eq_n(A)$ are trees.

For~$n=3$ and~$|A|\geq 2$, it is easy to see that the only minimal complexes generating~$\Eq_n(A)$ are the trees, because the language is irreducible and any pair of distinct positions~$i,j$ are independent.

However, the landscape starts to change for~$n\geq 4$ and~$|A|\geq 3$. We do not know whether it changes only from~$|A|=2$ to~$|A|=3$, or whether it keeps changing as the size of~$A$ grows. Observe that the family of complexes generating~$\Eq_n(A)$ can only decrease as~$A$ grows: if~$|A|\leq |A'|$ and a complex~$K$ generates~$\Eq_n(A')$, then~$K$ generates~$\Eq_n(A)$: one can compose any generation procedure for~$\Eq_n(A')$ with a surjective function~$f:A'\to A$ applied component-wise.

\subsubsection{Two counter-examples}
For~$n=4$ and~$|A|\geq 3$, we show that the following two trees do not generate~$L_4(A)$. At the end of Section \ref{sec_consec_sufficient}, we will explain why they are the only ones that do not generate~$L_4(A)$, up to symmetry by reflection~$i\mapsto 3-i$.

\begin{proposition}\label{prop_tree1}
The tree in Figure \ref{fig_ex1} does not generate~$L_4(A)$, if~$|A|\geq 3$.
\end{proposition}
\begin{figure}[ht]
\centering
\begin{tikzpicture}
    \node (a) [state, minimum size = 1mm] {$0$};
    \node (b) [state, right = 5mm of a, minimum size = 1mm] {$1$};
    \node (c) [state, right = 5mm of b, minimum size = 1mm] {$2$};
    \node (d) [state, right = 5mm of c, minimum size = 1mm] {$3$};

\path
    (a) edge [bend left=50] node [above] {$y$} (d)
    (b) edge [bend left] node[above] {$z$} (d)
    (a) edge [bend right=35] node[below] {$x$} (c);
    \end{tikzpicture}
\caption{A tree that does not generate~$L_4$ on three letters}\label{fig_ex1}
\end{figure}
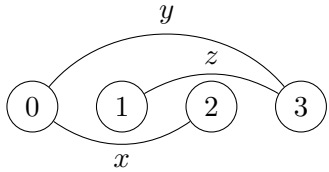
\begin{proof}
We can assume that~$A$ contains~$\{a,b,c\}$. Assume for a contradiction that~$F$ generates~$L_4(A)$ and~$K_F$ is contained in the tree.

Let~$xyz$ be such that~$F(xyz)=baac$. Let~$x',z'$ be such that~$F_2(x')=b$ and~$F_1(z')=c$. It implies that~$F(x'yz)=aabc$ and~$F(xyz')=bcaa$. Therefore~$F(x'yz')=acba\notin L_4(A)$, which is a contradiction.
\end{proof}

We give the second one.
\begin{proposition}\label{prop_tree2}
The tree in Figure \ref{fig_ex2} does not generate~$L_4(A)$, if~$|A|\geq 3$.
\end{proposition}
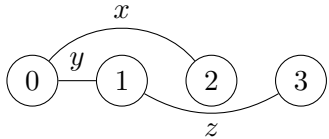
\begin{figure}[ht]
\centering
\begin{tikzpicture}
    \node (a) [state, minimum size = 1mm] {$0$};
    \node (b) [state, right = 5mm of a, minimum size = 1mm] {$1$};
    \node (c) [state, right = 5mm of b, minimum size = 1mm] {$2$};
    \node (d) [state, right = 5mm of c, minimum size = 1mm] {$3$};

\path
    (a) edge node [above] {$y$} (b)
    (a) edge [bend left=50] node[above] {$x$} (c)
    (b) edge [bend right] node[below] {$z$} (d);
    \end{tikzpicture}
\caption{This graph does not generate~$L_4$ on three letters}\label{fig_ex2}
\end{figure}
\begin{proof}
We can assume that~$A$ contains~$\{a,b,c\}$. Assume for a contradiction that~$F$ generates~$L_4(A)$ and~$K_F$ is contained in the tree.

Let~$xyz$ be such that~$F(xyz)=abaa$. Let~$x',z'$ be such that~$F_2(x')=c$ and~$F_3(z')=b$. One has~$F(x'yz)=bbca$, and~$F(xyz')=aaab$. Therefore,~$F(x'yz')=bacb\notin L_4(A)$.
\end{proof}

\subsubsection{A sufficient condition}\label{sec_consec_sufficient}
The same strategy as in Proposition \ref{prop_two_letters} can be applied to certain trees.
\begin{proposition} \label{prop_sufficient}
Let~$T$ be a rooted spanning tree such that each~$i\in I_n$ which is not a leaf has a descendant which is~$i-1$ or~$i+1$. Then~$T$ generates~$\Eq_n(A)$ for any~$A$. 
\end{proposition}
\begin{proof}
We apply the same strategy as in Proposition \ref{prop_two_letters}: if the edges that are incident to a vertex~$i$ do not have the same values, then~$i$ takes value~$w_j$, where~$j=i-1$ or~$i+1$ is a descendant of~$i$ and~$w$ is the edge starting from~$i$ and going towards~$j$.
\end{proof}

However, this condition is not necessary. Among the trees on~$I_4$, all of them satisfy the condition of Proposition \ref{prop_sufficient} except three: the ones from Propositions \ref{prop_tree1} and \ref{prop_tree2}, which do not generate~$L_4$, and the next one which does.
\begin{proposition}
The graph in Figure \ref{fig_graph2} generates~$L_4(A)$ for any~$A$.
\end{proposition}
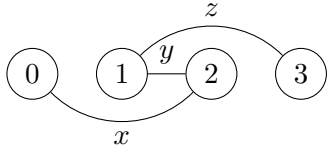
\begin{figure}[ht]
\centering
\begin{tikzpicture}
    \node (a) [state, minimum size = 1mm] {$0$};
    \node (b) [state, right = 5mm of a, minimum size = 1mm] {$1$};
    \node (c) [state, right = 5mm of b, minimum size = 1mm] {$2$};
    \node (d) [state, right = 5mm of c, minimum size = 1mm] {$3$};

\path
    (a) edge [bend right=45] node [below] {$x$} (c)
    (b) edge node[above] {$y$} (c)
    (b) edge [bend left=45] node[above] {$z$} (d);
    \end{tikzpicture}
\caption{A tree that generates~$L_4$ on three letters}\label{fig_graph2}
\end{figure}
Note that for every choice of a root,~$b$ or~$c$ is child of the other, and has no descendant which is a consecutive position.
\begin{proof}
The edges~$(0,2)$ and~$(1,3)$ take values in~$A$ and give their values to~$0$ and~$3$ respectively. The edge~$(1,2)$ takes values in~$A\times A\times \{1,2\}$. The idea is that such a value~$(u,v,w)$ gives possible values~$u$ and~$v$ to~$1$ and~$2$ respectively, but also chooses a vertex~$w\in\{1,2\}$ that is in charge of making sure that the output is correct (and the other vertex simply takes its assigned value). For instance if~$1$ is chosen, as it knows the value of~$3$ by reading~$z$, it takes its assigned value~$u$ only if~$2$ and~$3$ take the same value, otherwise it takes the same value as~$2$. If~$2$ is chosen, then the symmetric strategy is applied by~$2$, which sees the value of~$0$.

Precisely, the rules are:
\begin{align*}
F_0(x)&=x,\\
F_1(y,z)&=\begin{cases}
u&\text{if $y=(u,v,1)$ with~$v=z$},\\
v&\text{if $y=(u,v,1)$ with $v\neq z$},\\
u&\text{if }y=(u,v,2),
\end{cases}\\
F_2(x,y)&=\begin{cases}
v&\text{if $y=(u,v,2)$ with~$u=x$},\\
u&\text{if $y=(u,v,2)$ with $u\neq x$},\\
v&\text{if }y=(u,v,1),
\end{cases}\\
F_3(z)&=z.
\end{align*}

If~$y=(u,v,1)$ then
\[
F(x,y,z)=\begin{cases}
xuzz&\text{if~$v=z$},\\
xvvz&\text{if $v\neq z$},
\end{cases}
\]
and symmetrically, if~$y=(u,v,2)$ then
\[
F(x,y,z)=\begin{cases}
xxvz&\text{if~$u=x$},\\
xuuz&\text{if $u\neq x$}.
\end{cases}
\]
We see that all the outputs belong to~$L_4(A)$, and that every element of~$L_A(A)$ can be obtained.
\end{proof}

For~$n=4$, one can check that the trees that differ from the tree examples presented above all satisfy the condition of Proposition \ref{prop_sufficient}, so they generate~$L_4(A)$:
\begin{itemize}
\item If a tree has a vertex of degree~$3$ then it can be rooted in that vertex and satisfies the condition of Proposition \ref{prop_sufficient},
\item Otherwise it is a path~$(v_0,v_1,v_2,v_3)$. If~$v_0$ and~$v_1$ are consecutive then it can be rooted at~$v_1$, and symmetrically if~$v_2$ and~$v_3$ are consecutive then it can be rooted at~$v_2$: in both cases Proposition \ref{prop_sufficient} can be applied,
\item Otherwise, up to symmetry one has~$\{v_0,v_1\}=\{0,2\}$ and~$\{v_2,v_3\}=\{1,3\}$ so it is one of the three examples presented above.
\end{itemize}

\subsubsection{Paths must connect some consecutive positions}
It seems natural to expect that if a tree or a more general graph generates~$\Eq_n(A)$, then it should connect at least two consecutive positions~$i$ and~$i+1$. We are only able to prove it in the case of paths when the alphabet is sufficiently large. A path is a tree made of one root, one leaf and a branch between them. Whether the result holds for general trees is left as an open question.

\begin{theorem}
Assume that~$|A|\geq 7$. If a path generates~$\Eq_n(A)$, then it contains an edge between two consecutive positions.
\end{theorem}

%
%
%

We need to introduce some simple vocabulary. Let~$V$ be a subset of~$I_n$. A \textbf{$V$-sequence} is a function~$w:V\to A$. It is \textbf{valid} if there exists~$i$ such that both~$i$ and~$i+1$ belong to~$V$, and~$w(i)=w(i+1)$. In particular, if~$V$ does not contain any pair of consecutive positions, then there is no valid~$V$-sequence.

Ideally, we would like to remove the vertices one by one, so that at each step the remaining subpath has a generation procedure producing the valid partial sequences. We cannot quite achieve this, but almost.

\begin{definition}
For~$v\in V$, we say that a complex~$T$ over~$V$ is \textbf{$v$-good} if it is has a generation procedure that only produces valid~$V$-sequences and such that every~$(V\setminus\{v\})$-sequence is the restriction of some produced~$V$-sequence.
\end{definition}
If~$T$ over~$[0,n-1]$ generates~$\Eq_n(A)$, then~$T$ is~$v$-good for any~$v\in [0,n-1]$.

\begin{lemma}
Let~$v\in V$. If a complex~$T$ over~$V$ is~$v$-good, then~$v-1$ or~$v+1$ belongs to~$V$.
\end{lemma}
\begin{proof}
Take an invalid~$(V\setminus\{v\})$-sequence. It is the restriction of a produced~$V$-sequence; that produced~$V$-sequence is valid so it gives the same value to two consecutive positions. As its restriction to~$V\setminus\{v\}$ is invalid, one of the consecutive positions must be~$v$, so the other is~$v-1$ or~$v+1$.
\end{proof}

Let~$T$ be a tree and let~$v\in T$. The forest~$T\setminus \{v\}$ is a disjoint union of trees~$T_1,\ldots,T_k$, made of vertex sets~$V_1,\ldots,V_k$.
\begin{lemma}\label{lem_branch}
Assume that the tree~$T$ is~$v$-good over~$V$. If~$V_i$ contains neither~$v-1$ nor~$v+1$, then~$T\setminus V_i$ is~$v$-good over~$V\setminus V_i$.
\end{lemma}
\begin{proof}
Let~$J$ be the set of edges of~$T$ and~$f:B^J\to A^V$ be a function witnessing that~$T$ is~$v$-good. Let~$W=V\setminus (V_i\cup\{v\})$, let~$D=\bigcup_{u\in W}\W_f(u)$ be the corresponding set of inputs and let~$g:B^{D}\to A^W$ be the restriction of~$f$.

Let~$E$ be the set of edges that are incident to~$v$ in~$T\setminus V_i$ and~$\alpha$ a partial input whose domain is~$E$.

We will use the following simple and classical combinatorial observation. We include a proof for completeness.
\begin{lemma}\label{lem_combi}
Let~$A$ be a set~$(S_x)_{x\in X}$ be a family of subsets of~$A$ of size at most~$2$. If for every triplet~$(x_0,x_1,x_2)$, the set~$S_{x_0}\cap S_{x_1}\cap S_{x_2}$ is non-empty, then~$\bigcap_{x\in X}S_x$ is non-empty.
\end{lemma}
\begin{proof}
It can be shown by contraposition. Assume that~$\bigcap_{x\in X} S_x=\emptyset$ and let~$x_0$ be arbitrary. Let~$s_1,s_2\in S$ be such that~$S_{x_0}\subseteq \{s_1,s_2\}$. By assumption, for each~$k\in\{1,2\}$ there exists~$x_k\in X$ such that~$s_k\notin S_{x_k}$. Therefore,~$S_{x_0}\cap S_{x_1}\cap S_{x_2}=\emptyset$.
\end{proof}

\begin{claim}
There exists a letter~$a\in A$ such that for every total input~$\beta$ on~$D$ extending~$\alpha$,~$g(\beta)\cup [v\mapsto a]$ is valid.
\end{claim}
\begin{proof}
Note that when~$g(\beta)$ is already valid, the value of~$a$ does not matter, so we only need to consider the case when~$g(\beta)$ is invalid. Let
\[
X=\{\beta\in B^D:\beta\text{ extends $\alpha$ and~$g(\beta)$ is invalid.}\}
\]
For each~$\beta\in D$, let~$S_\beta$ be a subset of~$A$ of size at most~$2$ containing the values of~$g(\beta)$ at positions~$v-1$ and~$v+1$ (if they exist, i.e.~if~$v>0$ and~$v<n-1$ respectively). For any~$a\in S_\beta$, the~$(W\cup \{v\})$-sequence~$g(\beta)\cup [v\mapsto a]$ is valid. We need to show that~$\bigcap_{\beta\in X} S_\beta\neq \emptyset$. By Lemma \ref{lem_combi}, it is sufficient to show that for all triplets~$(\beta_0,\beta_1,\beta_2)\in X^3$,~$S_{\beta_0}\cap S_{\beta_1}\cap S_{\beta_2}\neq\emptyset$.
 

Let~$\beta_0,\beta_1,\beta_2$ extend~$\alpha$ and such that each~$w_k=g(\beta_k)$ is invalid. As~$|A|\geq 7$, there exists an invalid~$V_i$-sequence~$u$ such that each~$u\cup w_k$ is an invalid~$(V\setminus\{v\})$-sequence. Indeed, the letters in~$u$ can be successively chosen so that they do not coincide with the letters already present in consecutive positions in any~$u\cup w_k$: at each step, there are at most~$6$ such letters, so there is always a fresh letter available in~$A$. Let~$\gamma$ be an input on the edges incident to~$V_i$ and, sent to~$u$ by~$f$. Let~$a$ be the value of~$v$ in~$f(\gamma\cup\alpha)$. For each~$k$,~$f(\gamma\cup\beta_k)=u\cup w_k\cup [v\mapsto a]$ is valid but~$u\cup w_k$ is not, so an equality between consecutive positions must involve~$v$ and~$v-1$ or~$v+1$, so~$a\in S_{\beta_k}$. Therefore,~$S_{\beta_0}\cap S_{\beta_1}\cap S_{\beta_2}\neq \emptyset$.

By Lemma \ref{lem_combi}, we conclude that there exists~$a$ in~$\bigcap_{\beta\in X} S_{\beta}$, which proves the claim.
\end{proof}

To each~$\alpha$ one can therefore associate such an~$a(\alpha)$. We define a procedure on~$T\setminus V_i$. All the vertices except~$v$ keep the same local rule as in~$f$. The new rule for~$v$ is, upon reading the restriction~$\alpha$ of the input to~$E$, to output~$a(\alpha)$. The procedure only produces valid sequences, and every sequence on~$V\setminus (V_i\cup\{v\})$ is obtained (extend it to~$V\setminus \{v\}$, that sequence is obtained by the original procedure).
\end{proof}

\begin{lemma}\label{lem_uv}
Let~$v$ have degree~$1$ in the tree~$T$ and~$u$ its adjacent vertex. If~$T$ is~$v$-good, then~$T$ is~$u$-good.
\end{lemma}
\begin{proof}
We keep the same generation procedure, will only produces valid sequences. We need to show that every~$(V\setminus \{u\})$-sequence is the restriction of some produced~$V$-sequence.

First observe that~$v$ can take arbitrary values. Indeed, for~$a\in A$, take an invalid~$V\setminus\{v\}$-sequence assigning the value~$a$ to the neighbors of~$v$. It is the restriction of some produced sequence. That sequence is valid, so it must give value~$a$ to~$v$.

Let~$w$ be a~$V\setminus\{u\}$-sequence. Let~$w'$ be a~$V\setminus\{v\}$-sequence that coincides with~$w$ on~$V\setminus \{u,v\}$ (and gives an arbitrary value to~$u$). By assumption,~$w'$ is the restriction of a sequence that can be obtained from some input. In that input, change the value of the edge~$(u,v)$ to give value~$w(v)$ to~$v$. It may change the value of~$u$, but in any case the produced sequence extends~$w$.
\end{proof}

\begin{corollary}
If a path generates~$\Eq_n(A)$, then two neighbors must be connected by an edge.
\end{corollary}
\begin{proof}
Assume for a contradiction that a path~$T$ generates~$\Eq_n(A)$, with no edge between consecutive vertices. Let~$v_0,v_1,\ldots,v_{n-1}$ be the list of vertices in the order of the path.


Let~$V_i=\{v_i,\ldots,v_{n-1}\}$ and~$T_i$ be the subpath over~$V_i$, made of edges~$(v_j,v_{j+1})$ for~$i\leq j<n-1$.
\begin{claim}
For each~$i<n$,~$T_i$ is~$v_i$-good.
\end{claim}
\begin{proof}
We prove the claim by induction on~$i$.

For~$i=0$, as~$T_0$ generates~$\Eq_n(A)$,~$T_0$ is~$v_0$-good.

Assume that~$T_i$ is~$v_i$-good for some~$i<n-1$. As~$v_i$ has degree~$1$ in~$T_i$ and~$v_{i+1}$ is its adjacent vertex,~$T_i$ is~$v_{i+1}$-good by Lemma \ref{lem_uv}. The forest~$T_i\setminus \{v_{i+1}\}$ is made of two trees: the singleton~$\{v_i\}$ and the rest. By assumption,~$v_i$ is not consecutive to~$v_{i+1}$ so by Lemma \ref{lem_branch},~$T_i\setminus \{v_i\}=T_{i+1}$ is~$v_{i+1}$-good over~$V_i\setminus \{v_i\}=V_{i+1}$, which proves the induction step.
\end{proof}

As a result,~$T_{n-1}$ is~$v_{n-1}$-good, which is a contradiction as~$v_{n-1}$ as no consecutive position in~$V_{n-1}=\{v_{n-1}\}$.
\end{proof}

When trying to generalize this technique to trees, this technique stops working when~$T$ is~$v$-good and~$v-1$ and~$v+1$ lie in two different subtrees of~$T\setminus \{v\}$. One can remove all the other subtrees by Lemma \ref{lem_branch}, but then~$v$ has degree~$2$ in the remaining subtree so Lemma \ref{lem_uv} cannot be applied.


\section{Relationship with distributed computing}\label{sec_distributed}
The problem of generating a language in a local way can be seen as a task to be solved in the framework of distributed computing. However, there are key differences with the traditional problems in this area which make our problem slightly unusual. Let us mention the main differences: 
\begin{itemize}
\item Relationship between inputs and outputs. Most of the time in distributed computing, a task involves a particular relationship between the outputs and the inputs, described in the specification of the task. For instance, in the consensus task, processed are assigned inputs bits and they should agree on a common output bit with an important condition: if the input bits all have a common value, then the output also takes that value. In our setting, we are only interested in the set of possible outputs, and not on their relationship with the inputs.
\item Explicit communication. In the usual models, communication is allowed either by message-passing, or by shared read-write memory. In both cases, the processes write pieces of information to be read by other processes, which is only possible if the protocol works in steps, often called layers. In our case, communication is only implicit and made possible by sharing parts of the input, and the whole computation is made in one step.
\item Distributed computing often focuses on the resilience of protocols to crash failures, to asynchronicity. In our setting, no unpredictable behavior is allowed.
\end{itemize}

Combinatorial topology has been an important way of reformulating problems in distributed computing, providing topological techniques to prove impossibility results \cite{HKR13}. We briefly recall the general idea and then show how the problem of language generation can be reformulated in this setting.

\subsection{Distributed computing via combinatorial topology}
The general problem in distributed computing is that several processes are required to solve a task together by collaboratively processing a global state though having only local views on this state. In the combinatorial topology formulation, the set of all possible configurations is represented by a chromatic labeled simplicial complex in which each simplex represents a possible global state, and its vertices represent the local views of the processes on that state.

More precisely, each process has a unique identifier (which is an element of a set~$C$ of colors) and a local view on the global state (which is an element of a set~$V$ of values, or labels). The set of all possible local views of the processes is a subset of~$C\times V$, seen as a vertex set~$V$. A global state is represented by a simplex whose vertices represent the local views of the processes on this state, so it contains exactly one vertex of each color, whose label is its local view. Computations are then represented by chromatic simplicial maps, which send vertices to vertices of the same color, and simplices to simplices.

This framework has enabled the application of topological techniques to prove impossibility results in distributed computing. For instance, Sperner's lemma, which is a combinatorial formulation of Brouwer's fixed-point theorem, has been applied to show that the so-called set agreement task is not solvable in certain models.



\subsection{Language generation via combinatorial topology}
We now explain how the problem of language generation can be formulated in this language.


Let~$L\subseteq A^I$ and~$K$ be a simplicial complex over~$I$. Let~$f:B^J\to A^I$ be a hypothetical function such that~$\im{f}=L$ and~$K_f\subseteq K$. We can assume that~$J$ is the set of maximal simplices of~$K$.

A language~$L\subseteq A^I$ induces an \textbf{output complex} $\O_L$, which is a chromatic labeled simplicial complex with~$I$ as color set and~$A$ as label set. 
For each element~$y\in L$, there is a simplex~$S=\{(i,y_i):i\in I\}$ and the vertex set is the set of pairs~$(i,a)\in I\times A$ appearing in these simplices. In other words, each simplex represents an element of~$L$, its vertices represent its values at each position.

Let~$B$ be a finite set and~$B_\bot=B\sqcup\{\bot\}$ and let~$K$ be a simplicial complex~$K$ over~$I$. Let~$J$ be the set of maximal simplices of~$K$. We define the the \textbf{input complex} $\II_K(B)$ which is a chromatic labeled simplicial complex with~$I$ as color set and~$B_\bot^J$ as label set. 
For each~$x\in B^J$, there is a simplex~$\{(i,\restr{x}{\W_f(i)})\}$, where the partial sequence~$\restr{x}{\W_f(i)}$ can be seen as an element of~$B_\bot^J$.

As explained above, we assume no relationship between the inputs and the outputs, so we do not need a carrier map as in the usual setting, which we therefore do not redefine.

Our goal is then to find a chromatic simplicial map~$F:\II_K(B)\to\O_L$ which is surjective in the sense that every output simplex is the image of some input simplex. More precisely,~$F$ is a function from the vertices of~$\II_K(B)$ to the vertices of~$\O_L$ that is color-preserving (i.e.~is the identity on the first component) and sends each input simplex to an output simplex, i.e.~each~$x\in B^J$ to an element~$y\in L$. We are requiring moreover that the map is surjective, i.e.~each simplex~$y\in L$ is the image of some simplex~$x\in B^J$.

\begin{proposition}
A complex~$K$ generates a language~$L$ if and only if there exists~$B$ and a surjective chromatic simplicial map~$F:\II_K(B)\to\O_L$.
\end{proposition}
\begin{proof}
We know from Proposition \ref{prop_canonical} that~$K$ generates~$L$ if and only if there exists~$B$ and~$F:B^J\to A^I$ such that~$\im{F}=L$ and~$K_F\subseteq K$. Such a map induces a chromatic simplicial map~$G:I\times B_\bot\times I\times A$ sending~$(i,x)$ to~$(i,F_i(x))$. Conversely, any such~$G$ induces a map from the simplices of~$\II_K(B)\to \O_L$, which is a function~$F:B^J\to L$ such that~$K_F\subseteq K$.
\end{proof}

Although this reformulation opens the possibility of applying techniques from combinatorial topology to our problem, it does not seem to provide new insight or simpler proofs for the results of this article. The geometrical intuition that could be gained by this translation quickly becomes elusive because the number of vertices and the dimension of the input and output complexes grow with the length of the strings. Whether topological techniques such as Sperner's lemma can be applied to our problem is left as an open question.

\subsection{Examples}
Let us illustrate the reformulation of the language generation problem on a few examples.

Let~$n=3$. We consider two languages in~$\{0,1\}^3$: the set~$\card_{\leq 1}$ of sequences having at most one occurrence of~$1$, and the set~$\U_3$ of sequences having exactly one occurrence of~$1$. We also consider two communication complexes over~$I_3$: the complete graph~$K_3$ and the graph~$G$ obtained from~$K_3$ by removing an edge. Table \ref{tab_gen} summarizes which complexes generate which languages.
\begin{table}[ht]
\[
\begin{array}{c|c|c}
&\card_{\leq 1}&\U_3\\\hline
K_3&\text{yes}&\text{no}\\\hline
G&\text{no}&\text{no}
\end{array}
\]
\caption{The table shows which communication complexes generate which languages}\label{tab_gen}
\end{table}

The corresponding output complexes~$\O_{\card_{\leq 1}}$ and~$\O_{\U_3}$ are shown in Figure \ref{fig_output_complex}, where the set~$I_3=\{0,1,2\}$ is identified with the color set~$\{\pic{chromatic-100},\raisebox{-.5mm}{\includegraphics{Images/chromatic-101}},\raisebox{-.5mm}{\includegraphics{Images/chromatic-102}}\}$, and the labels in~$\{0,1\}$ are written on the vertices.

\begin{figure}[ht]
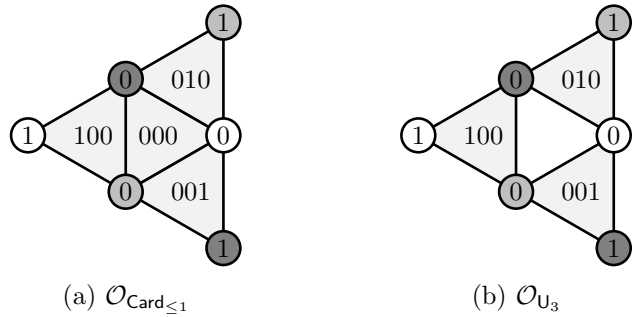

\centering
\subcaptionbox{$\O_{\card_{\leq 1}}$\label{fig_output_complex_card}}[.4\textwidth]{\includegraphics{Images/chromatic-0}}
\subcaptionbox{$\O_{\U_3}$\label{fig_output_complex_unique}}[.4\textwidth]{\includegraphics{Images/chromatic-3}}
\caption{The output complexes associated with two languages in~$\{0,1\}^3$}\label{fig_output_complex}
\end{figure}


The communication complexes~$K_3$ and~$G$ together with their associated input complexes with~$B=\{0,1\}$ are shown in Figures \ref{fig_clique_and_input} and \ref{fig_path_and_input}. In each case, the maximal simplices of~$K_3$ and~$G$ are edges and the lengths of the labels of the vertices in the input complexes correspond to the number of edges (note that the input complex of~$K_3$ is the same complex as in the example of the muddy children problem \cite[Section 1.3.1]{HKR13}).

\begin{figure}[ht]
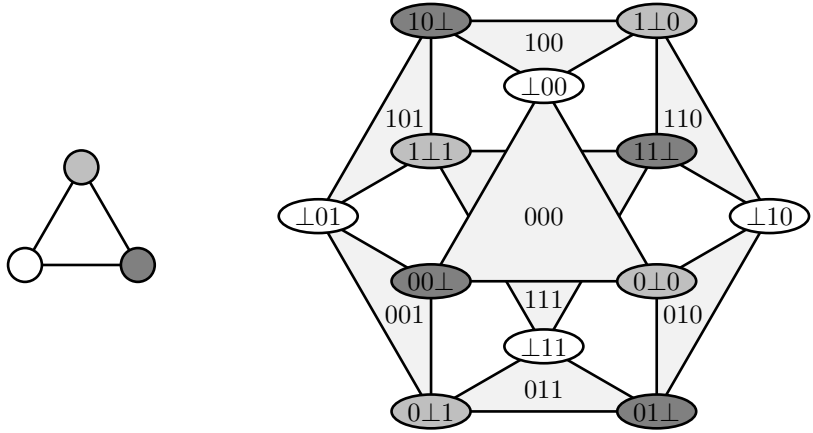

\subcaptionbox{The complete graph~$K_3$\label{fig_clique}}[.35\textwidth]{\includegraphics{Images/chromatic-20}}
\subcaptionbox{The input complex~$\II_{K_3}$ with~$B=\{0,1\}$\label{fig_input_clique}}[.6\textwidth]{\includegraphics{Images/chromatic-2}}
\caption{A communication complex and the associated input complex}\label{fig_clique_and_input}
\end{figure}
\begin{figure}[ht]
\centering
\subcaptionbox{A graph~$G$\label{fig_path}}[.35\textwidth]{\includegraphics{Images/chromatic-10}}
\subcaptionbox{The input complex~$\II_G$ with~$B=\{0,1\}$\label{fig_input_path}}[.6\textwidth]{\includegraphics{Images/chromatic-1}}
\caption{A communication complex and the associated input complex}\label{fig_path_and_input}
\end{figure}

As explained in Remark \ref{rmk_at_most1}, the complete graph~$K_3$ generates~$\card_{\leq 1}$, moreover there is a generation procedure with input alphabet~$B=\{0,1\}$. In other words, there is a surjective chromatic simplicial map from~$\II_{K_3}$ to~$\O_{\card_{\leq 1}}$: it sends the vertices \pic[-1mm]{chromatic-300}, \pic[-1mm]{chromatic-301} and \pic[-1mm]{chromatic-302} to the vertices of the corresponding colors labeled by~$1$, and the other vertices to the vertices labeled by~$0$ (see \cite[GL]{GL} for a visualization).

However, the graph~$G$ does not generate~$\card_{\leq 1}$. In other words, there is no surjective chromatic simplicial map from~$\II_G$ (for any~$B$) to~$\O_{\card_{\leq 1}}$. For~$B=\{0,1\}$, it can be easily seen from the pictures: in~$\II_G$, every pair of white and black vertices are joined by an edge, which is not the case in~$\O_{\card_{\leq 1}}$ (it corresponds to the fact that the white and black vertices are not connected by an edge in~$G$, but are not independent w.r.t.~$L$).

We know from Proposition \ref{prop_unique_5} that the only complex generating~$\U_3$ is the full complex. In particular, the complete graph~$K_3$ over~$I_3$ does not generate~$\U_3$ hence there is no surjective chromatic simplicial map from~$\II_{K_3}$ to~$\O_{\U_3}$.

\subsection{Some results revisited}
Some of the simplest concepts and results of the article can be translated in the setting of chromatic labeled simplicial complexes.

First, joins of input and output complexes reflects simple properties of communication complexes and languages:
\begin{itemize}
\item A language~$L$ is not irreducible, i.e.~is a product~$L_1\times L_2$ of two languages, if and only if the associated output complex~$\O_L$ is the join of~$\O_{L_1}$ and~$\O_{L_2}$,
\item A complex~$K$ is disconnected, i.e.~is the disjoint union of two complexes~$K_1$ and~$K_2$, if and only if the associated input complex~$\II_K$ is the join of~$\II_{K_1}$ and~$\II_{K_2}$.
\end{itemize}
Therefore, Proposition \ref{prop_irreducible} which states in particular that an irreducible language cannot be generated by a disconnected complex, can be reformulated as follows: a simplicial map sends the join of two input complexes to the join of their images, so a disconnected complex can only generate a non-irreducible language.

For a communication complex~$K$, the connectedness of its input complex~$\II_K$ reflects that~$K$ is not the full complex:
\begin{itemize}
\item If~$K$ is the full complex over~$I$, then~$\II_K(B)$ is a disjoint union of simplices, one for each element of~$B$, in particular it is disconnected if~$|B|\geq 2$,
\item If~$K$ is a complex over~$I$ and is not the full complex, then~$\II_K(B)$ is connected. The proof is similar to Proposition \ref{prop_well_defined}. Let~$J$ be the set of maximal simplices of~$K$ and for~$i\in I$, let~$\W(i)\subseteq J$ be the set of maximal simplices of~$K$ containing~$i$. Let~$x,y\in B^J$, which can be seen as simplices of~$\II_K(B)$. Define~$x=x_0,x_1,\ldots,x_n$ where~$x_{i+1}$ coincides with~$x_i$ on~$\W(i)$ and with~$y$ elsewhere. As~$\bigcap_i\W(i)=\emptyset$, one has~$x_n=y$. The simplices~$x_i$ and~$x_{i+1}$ of~$\II_K(B)$ share a vertex~$(i,\restr{x_i}{\W(i)})$, so the simplices~$x$ and~$y$ in~$\II_K(B)$ belong to the same connected component.
\end{itemize}
Therefore, Example \ref{ex_constant} which states that the only complex generating the language of constant sequences, is the full complex, can be derived in the following way. Let~$L$ be the language of constant sequences over~$A$, with~$|A|\geq 2$. Its output complex~$\O_L$ is the disjoint union of the simplices~$\{(i,a):i\in I\}$ for~$a\in A$, in particular it is disconnected. If~$K$ generates~$L$, then~$\II_K$ must be disconnected as well, implying that~$K$ is the full complex. 

%
%

It is not clear whether more complicated results are easier to obtain using this framework.

\section{Future directions}\label{sec_conclusion}
This study leaves many open questions, let us list a few ones:
\begin{itemize}
\item What are the minimal complexes generating~$\U_n$ for arbitrary~$n$?
\item What are the minimal complexes generating~$\Eq_n(A)$ for arbitrary~$n$ and~$A$?
\item Is there a characterization of the languages that can be generated by a complex which is not the full simplex, i.e.~by a non-trivial procedure in which the cells do not communicate all together. For~$n=2$, those languages are the non-irreducible ones, and the situation is already not clear for~$n=3$.
\item Is there a characterization of the languages that can be generated by a graph, i.e.~by a procedure in which cells only communicate in pairs?
\item Given finite sets~$A,I$, a language~$L\subseteq A^I$ and a complex~$K$ over~$I$, what is the computational complexity of deciding whether~$K$ generates~$L$? It is decidable because it can be answered by exhaustive search in a finite space (see Proposition \ref{prop_canonical} implies that the input space of a potential generating function can be chosen in advance). Is there a language~$L$ for which this problem is as difficult as the general case?
\item Are the topological properties of complexes generating a language reflected in intrinsic properties of the language? For instance, we have seen with Proposition \ref{prop_irreducible} that the connectedness of complexes reflects the irreducibility of the language. Are there other correspondences, involving for instance higher-dimensional connectedness notions from algebraic topology?
\end{itemize}

In a sequel article \cite{H25b}, we thoroughly investigate the language of binary monotonic sequences, whose analysis turns out to be much harder than the languages studied in the present article.

Finally, when proving that a complex does not generate a language, some of the arguments follow a common pattern: they consist in progressively deriving constraints that eventually lead to a conflict (see for instance Section \ref{sec_non_dec}). This type of arguments can be presented in a unified framework, for instance in the form of an inference system, and the search for such arguments can be partially automated. It can help analyzing certain languages but quickly reaches its limits due to the combinatorial explosion of the search space.


\end{document}